\documentclass{article}
\usepackage{ijcai19}

\newif\ifsup\suptrue 

\def\nocomment{1}

\usepackage{times}
\usepackage{nimbusmononarrow}
\usepackage{soul}
\usepackage{url}
\usepackage[hidelinks]{hyperref}
\usepackage[utf8]{inputenc}
\usepackage[small]{caption}
\usepackage[table]{xcolor}
\usepackage{algorithm}
\usepackage{graphicx}
\usepackage{amsmath}
\usepackage{amsthm}
\usepackage{bm} 
\usepackage{amssymb}
\usepackage{tikz}
\usetikzlibrary{arrows}
\usetikzlibrary{decorations.pathreplacing}
\usetikzlibrary{automata,positioning}
\usepackage{booktabs} 
\usepackage{wrapfig}

\definecolor{dkblue}{cmyk}{1,.54,.04,.19}
\hypersetup{
      unicode=false,          
      pdftoolbar=true,        
      pdfmenubar=true,        
      pdffitwindow=false,     
      pdfstartview={FitH},    
      pdfcreator={pdflatex},   
      pdfproducer={Producer}, 
      pdfnewwindow=true,      
      colorlinks=true,        
      linkcolor=black,       
      citecolor=black,       
      filecolor=black,       
      urlcolor=black,        
}

\newcommand{\floor}[1]{\left \lfloor {#1} \right\rfloor}
\newcommand{\ceil}[1]{\left \lceil {#1} \right\rceil}

\newcommand{\tceil}[1]{A_{\scriptscriptstyle{>}}(#1)}
\newcommand{\tfloor}[1]{A_{\scriptscriptstyle{\leq}}(#1)}

\newcommand{\low}{{\scriptstyle\texttt{low}}}
\newcommand{\high}{{\scriptstyle\texttt{high}}}

\newcommand{\crit}{{\scriptstyle\texttt{crit}}}

\newcommand{\used}{{\scriptstyle\texttt{used}}}

\newcommand{\succf}{\texttt{succ}}

\newcommand{\fgap}{\delta}
\newcommand{\fgapmin}{\fgap_\text{min}}

\newcommand{\cost}{C} 
\newcommand{\costlow}{\cost_{\low}} 
\newcommand{\costhigh}{\cost_{\high}} 
\newcommand{\costmin}{\cost_{\min}} 
\newcommand{\costcrit}{\cost_{\crit}} 

\newcommand{\sinit}{s_\textup{init}}

\newcommand{\nstar}{n_*}
\newcommand{\nstartree}{n_{T*}}
\newcommand{\nstargraph}{n_{G*}}
\newcommand{\coststar}{\cost^*}
\newcommand{\nbisec}{n_{\scriptstyle{\texttt{exp}}}}

\newcommand{\node}{n}

\newcommand{\cS}{\mathcal S}

\newcommand{\ibbslong}{Iterative Budgeted Exponential Search}
\newcommand{\ibbs}{IBEX}
\newcommand{\ibbstslong}{Budgeted Tree Search}
\newcommand{\ibbsts}{BTS}
\newcommand{\ibbsgslong}{Budgeted Graph Search}
\newcommand{\ibbsgs}{BGS}
\newcommand{\dovibbslong}{Dovetailing \ibbs}
\newcommand{\dovibbs}{Dov\ibbs}
\newcommand{\dovibbsgs}{Dov\ibbsgs}
\newcommand{\dovibbsts}{Dov\ibbsts}
\newcommand{\ubslong}{Uniform Budgeted Scheduler}
\newcommand{\ubs}{\text{UBS}}

\newcommand{\eda}{EDA*}
\newcommand{\idacr}{IDA*\_CR}
\newcommand{\ida}{IDA*}

\newcommand{\budget}{b}
\newcommand{\query}{\textcolor{lstfuncolor}{\texttt{query}}}
\newcommand{\queryp}{\textcolor{lstfuncolor}{\query^+}}
\newcommand{\querylim}{\textcolor{lstfuncolor}{\query_{\text{lim}}}} 
\newcommand{\queryext}{\textcolor{lstfuncolor}{\query_{\text{ext}}}} 
\newcommand{\queryextp}{\textcolor{lstfuncolor}{\query_{\text{ext}}^+}} 
\newcommand{\bisection}{\textcolor{lstfuncolor}{\texttt{exp\_search}}}

\newcommand{\queryint}{\textcolor{lstfuncolor}{\query_{\text{int}}}} 

\newcommand{\exptext}[1]{\text{#1}}
\newcommand{\state}{\texttt{state}} 
\usepackage{enumitem}
\usepackage{listings} 
\usepackage{mathtools}
\usepackage[capitalise]{cleveref}
\crefname{listing}{Algorithm}{Algorithms}
\Crefname{listing}{Algorithm}{Algorithms}

\theoremstyle{definition}
\newtheorem{theorem}{Theorem}

\newtheorem{proposition}[theorem]{Proposition}
\newtheorem{corollary}[theorem]{Corollary}

\newtheorem{remark}[theorem]{Remark}
\newtheorem{example}[theorem]{Example}

\definecolor{darkred}{rgb}{.7,0,0}
\definecolor{darkgreen}{rgb}{0,.5,0}
\definecolor{darkblue}{rgb}{0,0,.8}
\definecolor{darkcyan}{rgb}{0,0.6,.6}
\definecolor{darkorange}{rgb}{.8,.4,0}
\definecolor{gray}{rgb}{.4,.4,.4}

\ifdefined\nocomment{}
\newcommand{\todo}[1]{}
\newcommand{\comment}[1]{}
\newcommand{\warning}[1]{}
\newcommand{\quest}[1]{}

\newcommand{\todot}[1]{}
\else
\newcommand{\todo}[1]{\textcolor{darkorange}{(\emph{TODO: #1})}}
\newcommand{\comment}[1]{\textcolor{darkblue}{(\emph{#1})}}
\newcommand{\warning}[1]{\textcolor{red}{(\emph{WARNING: #1})}}
\newcommand{\quest}[1]{\textcolor{darkgreen}{(\emph{Q: #1})}}

\newcommand{\todot}[1]{{\tiny \textcolor{darkorange}{(\emph{TOR: #1})}}}
\fi

\newcommand{\Naturals}{\mathbb{N}_1}
\newcommand{\Nonnegints}{\mathbb{N}_0}

\DeclareMathOperator*{\argmin}{\text{argmin}}

\newcommand{\eps}{\varepsilon}

\newcommand{\citet}[1]{\citeauthor{#1}~[\citeyear{#1}]}

\ifsup
\definecolor{lstcommentcolor}{rgb}{.8,.4,0}
\definecolor{lstfuncolor}{rgb}{0,0,0} 
\definecolor{lstfuncolorr}{rgb}{.7,0,0} 
\else
\definecolor{lstcommentcolor}{rgb}{0,0,0}
\definecolor{lstfuncolor}{rgb}{0,0,0} 
\definecolor{lstfuncolorr}{rgb}{0,0,0} 
\fi

\newcommand{\lstmathcomment}[1]{\textcolor{lstcommentcolor}{#1}}
\newcommand*{\code}{\lstinline[keywordstyle=\color{lstfuncolor}, basicstyle=\color{lstfuncolor},classoffset=1,keywordstyle=\color{lstfuncolor}]}

\tikzstyle{ya}=[anchor=east,font=\scriptsize]
\tikzstyle{xa}=[anchor=north,font=\scriptsize]
\tikzstyle{axis} = [-latex]
\tikzstyle{line} = [thick]
\tikzstyle{end} = [fill]

\title{\ibbslong}

\author{
Malte Helmert\footnotemark$^1$
\and
Tor Lattimore$^2$\and
Levi H. S. Lelis$^3$\and
Laurent Orseau$^2$ \textnormal{and} 
Nathan R. Sturtevant$^4$\\
\affiliations
$^1$Department of Mathematics and Computer Science, University of Basel, Switzerland\\
$^2$DeepMind, London, UK\\
$^3$Departamento de Inform\'{a}tica, Universidade Federal de Vi\c{c}osa, Brazil\\
$^4$Department of Computing Science, University of Alberta, Edmonton, AB, Canada\\
\emails
malte.helmert@unibas.ch,
\{lattimore, lorseau\}@google.com,
levi.lelis@ufv.br,
nathanst@ualberta.ca
}

\begin{document}

\maketitle

\begin{abstract}
\ifsup
We tackle two long-standing problems related to re-expansions in heuristic search algorithms. For graph search,
A* can require $\Omega(2^{\nstar})$ expansions, where $\nstar$ is the number of states within the final $f$ bound. Existing algorithms that address this problem like B and B' improve this bound to $\Omega(\nstar^2)$. For tree search, IDA* can also require $\Omega(\nstar^2)$ expansions. We describe a new algorithmic framework that iteratively controls an expansion budget and solution cost limit, giving rise to new graph and tree search algorithms for which the number of expansions is 
$O(\nstar \log \coststar)$, where $\coststar$ is the optimal solution cost. Our experiments show that
the new algorithms are robust in scenarios where existing algorithms fail. In the case of tree search,
our new algorithms have no overhead over IDA* in scenarios to which IDA* is well suited and can therefore be recommended as a general replacement for IDA*.
\fi\end{abstract}

\renewcommand{\thefootnote}{\fnsymbol{footnote}}
\setcounter{footnote}{1}
\footnotetext{Alphabetical order. This paper is the result of merging two independent submissions to IJCAI 2019 \cite{orseau2019zooming,sturtevant2019exponential}.}
\renewcommand{\thefootnote}{\arabic{footnote}}
\setcounter{footnote}{0}

\section{Introduction}

There are two long-standing problems in heuristic search where existing algorithms 
struggle 
to balance the number of expansions and re-expansions performed in comparison to an oracle. One is in graph search, the other in tree search.

The first problem deals with admissible but inconsistent heuristics in \emph{graph search}.
With some caveats \cite{holte2010misconceptions}, A* with an admissible and consistent heuristic 
expands the minimum required number of states \cite{hart1968astar,dechter1985generalized}. However, with inconsistent heuristics it may expand exponentially more states than more cautious algorithms such as 
B \cite{Mar77} and B' \cite{mero1984pathmax},
which have a quadratic worst case.

The second problem is in heuristic \emph{tree search} algorithms that use memory that grows only linearly with the search depth.
In contrast, A* memory usage
grows linearly with time and often exponentially with the depth of the search.
To satisfy such low memory requirements, linear-memory tree search algorithms perform successive depth-first searches
with an increasing limit on the cost.
They also forgo global duplicate elimination, meaning that they do not detect if multiple paths from the initial state lead to the same state, which can lead to exponentially worse runtime compared to algorithms like A* when such duplicates are frequent. Hybrid algorithms that uses bounded memory for duplicate elimination are possible \cite[for example]{akagi2010transpositions}.

IDA* \cite{korf1985ida} is a cautious linear-memory algorithm that increases the f-cost bound minimally
(see also RBFS \cite{1993rbfskorf}).
At each iteration, IDA* searches all nodes up to the f-cost bound.
The minimum cost of the nodes pruned in one iteration
becomes the cost bound for the next iteration.
This approach ensures that the last cost bound will be exactly the minimum solution cost.
This is efficient when the number of nodes matching the
current cost bound grows exponentially with the number
of iterations,
as the total number of expansions will be dominated by the last iteration.
In the worst case however, IDA* may expand only one new node in each iteration, leading to a quadratic number of (re-)expansions.
Several methods have been developed to mitigate the re-expansion overhead of IDA* \cite{burns2013idaim,sharon2014eda,hatem2018solving,sarkar1991reducing,wah1994rida}
by increasing the cost bound more aggressively at each iteration
with the aim of achieving an exponential growth rate.
However, with this approach the last cost bound can be larger than the minimum solution cost, which may incur an arbitrarily large performance penalty.
For these algorithms, theoretical guarantees (when provided) require strong assumptions such as uniformity of the costs or branching factor \cite[for example]{hatem2015rbfscr}.


We propose a novel framework called \ibbslong{} (\ibbs) 
guaranteeing for both problems described above that the number of expansions 
is near-linear in the number of nodes whose cost is at most the minimum solution cost.
This is achieved by combining two ideas: (1) a budget on the number of expansions and (2) an exponential search for the maximum f-cost that can be searched exhaustively within the given budget.
This framework proves that no solution can be found for the current budget, 
and then doubles it until a solution is found.
This ensures that the last budget is always within twice the minimum required budget,
while amortizing the work on early iterations due to the exponential growth of the budget.

We develop two simple and fast algorithms that enjoy near-linear expansion guarantees, propose a number of enhancements,
show how the tree search and graph search problems can be reduced to our framework,
and show that these algorithms perform at least as well as state-of-the-art algorithms on a number of
traditional domains without exhibiting any of the catastrophic failure cases.

\section{Heuristic Search Problems}\label{sec:search}

A black box heuristic search problem is defined by a finite state space $\cS$, a set of goal states $\cS^* \subseteq \cS$,
a cost function $c : \cS \times \cS \to [0,\infty]$ and an initial state $\sinit \in \cS$.
The successors of a state $s$ are those states $s'$ that can be reached by a finite-cost edge: $\succf(s) = \{s' : c(s, s') < \infty\}$.
This defines a directed graph $\mathcal{G} = (\cS, E)$ where states correspond to vertices in the graph and the edges are the finite-cost successors $E = \{(s, s') : c(s, s') < \infty\}$.
A path is a sequence of states $\pi = (s_t)_{t=1}^m$ with $s_1 = \sinit$ and its cost is
$g(\pi) = \sum_{t=1}^{m-1} c(s_t, s_{t+1})$,
which may be infinite if there is no edge between adjacent states. 
The end state of path $\pi = (s_t)_{t=1}^m$ is
$\state(\pi) = s_m$ and its successor paths are $\succf(\pi) =  \{(s_t)_{t=1}^{m+1} : s_{t+1} \in \succf(s_m)\}$.
A state $s$ is expanded when the function $\succf(s)$ is called for $s$; a state $s'$ is generated when $\succf(s)$ is called creating $s' \in \succf(s)$. 
A node 
corresponds to a single path $\pi$ from the root of a search tree.
Expanding a node corresponds to expanding $\state(\pi)$ and generating all the corresponding successor paths (nodes). 
Search algorithms may expand the same state multiple times because multiple nodes might represent the same state.
Let $\pi^*(s)=\argmin_{\pi: \state(\pi)=s} g(\pi)$ be a least-cost path to state $s$ and $g^*(s) = g(\pi^*(s))$ be the cost of such a path. 
Let $\Pi^*$ be the set of all paths from the initial state to all goal states. Then, the cost of a least-cost path to a goal state is
$\coststar = \min_{\pi \in \Pi^*} g(\pi)$.
The objective is to find a least-cost path from the initial state to a goal state. 

Let $h^*(s)$ be the minimal cost over all paths from $s$ to any goal state.
A heuristic is a function $h : \cS \to [0,\infty]$ that provides an estimate of $h^*$.
A heuristic is admissible if $h(s) \leq h^*(s)$ for all states $s \in \cS$ and consistent if
$h(s) \leq h(s') + c(s, s')$ for all pairs of states $s, s'$.
The $f$-cost of a path is $f(\pi) = g(\pi) + h(\state(\pi))$.
Note that if $s=\state(\pi)$ is a goal state and the heuristic is admissible we must have $h(s)=0$.

We say that a search algorithm is a \emph{graph search} if it
eliminates duplicates of states generated by the algorithm;
otherwise it is called a \emph{tree search}.

\subsection{Graph Search}
With a consistent heuristic, $f$-costs along a path are non-decreasing, thus a graph search algorithm must expand all states in the graph with $f(s) = g^*(s)+h(s) < \coststar$.
In this setting, A* has an optimal behaviour \cite{dechter1985generalized}.
When the heuristic is admissible but inconsistent, 
for comparing algorithms
one could consider the ideal number of nodes that A* would expand if the heuristic
was made consistent.
Unfortunately, 
not only does there exist no optimal algorithm for this case \cite{mero1984pathmax},
but it can even be shown that \emph{all} algorithms may need to expand exponentially too many nodes
in some cases (see \ifsup{}\cref{sec:intract}\fi).
Hence we focus our attention on the following relaxed notion of optimality.
Let $\nstargraph = |\{s : \min_{\pi:\state(\pi) = s}\max_{s' \in \pi(s)} f(s') \leq \coststar\}|$ 
be the number of states that can be reached by a path along which all states
have $f$-cost at most $\coststar$;
this is the definition used by \citet{Mar77}.
\todo{check this thoroughly}
Then, there exist problems where A* performs up to $\Omega(2^{\nstargraph})$ expansions \cite{Mar77}.
This limitation has been partially addressed with the B~\cite{Mar77} and B'~\cite{mero1984pathmax} algorithms for which the number of expansions is at most $O(\nstargraph^2)$.
We improve on this result with a new algorithm for which the number of expansions is at most $O(\nstargraph \log(\coststar))$.

\subsection{Tree Search}

Tree search algorithms work on the tree expansion of the state space, where every path from $\sinit$ corresponds to a tree node. Consequently, states reached on multiple paths will be expanded multiple times.

We say a path $\pi$ is
necessarily expanded if $\max_{s' \in \pi} f(s') < \coststar$ and
possibly expanded if $\max_{s' \in \pi} f(s') \leq \coststar$.
A tree search algorithm must always expand all necessarily expanded paths
and will usually also expand some paths that are possibly but not necessarily expanded.
To avoid the subtleties of tie-breaking, we discuss upper bounds in terms of possibly expanded paths,
leaving a more detailed analysis for future work.
We write $\nstartree$ for the number of possibly expanded paths in a tree search.

In the worst case, IDA* may perform $\Omega(\nstartree^2)$ expansions.
To mitigate this issue, algorithms such as IDA*\_CR~\cite{sarkar1991reducing} and EDA*~\cite{sharon2014eda} increase the $f$-cost bound more aggressively.
These methods are effective when the growth of the tree is regular enough, but can fail catastrophically when the tree grows rapidly near the optimal $f$-cost, as will be observed in the experiments.
We provide an algorithm that performs at most a logarithmic factor more expansions than $\nstartree$ 
and uses memory that is linear in the search depth.

\paragraph{Notation.}
The natural numbers are $\Nonnegints = \{0,1,2,\ldots\}$ and $\Naturals = \{1,2,3,\ldots\}$.
For real-valued $x$ and $a$ let $\ceil{x}_{\geq a} = \max\{a, \ceil{x}\}$ and similarly for $\floor{x}_{\geq a}$.

\section{Abstract View}

We now introduce a useful abstraction that allows us to treat tree and graph search in a unified manner. Tree search is used as a motivating example.
The problem with algorithms like EDA* that aggressively increase the $f$-cost limit is the possibility of a significant number of wasted expansions once the $f$-cost limit is above $\coststar$. The core insight of our framework is that this can be mitigated by stopping the search if the number of expansions exceeds a budget and slowly increasing the budget in a careful manner. 

A depth-first search with an $f$-cost limit and expansion budget reveals that either (a) the expansion budget was insufficient to search the whole tree with $f$-cost smaller or equal to the limit, or (b) the expansion budget was sufficient. 
In the latter case, if the goal is found, then the algorithm can return a certifiably optimal solution.
Furthermore, when the budget is insufficient the largest $f$-cost of a node visited by the search serves as an upper bound on the largest $f$-cost for which the budget will be exceeded. When the budget is sufficient, the smallest $f$-cost in the fringe is a lower bound on the same. 
This information means that combining exponential search \cite{bentley1976expsearch} with repeated depth-first searches with a varying $f$-cost limit and fixed expansion budget can be used to quickly find a solution if the budget is sufficient to expand all nodes with $f$-cost less than $C^*$ and otherwise produce a certificate that the budget is insufficient, a process we explain in detail in \cref{sec:exp_search}.

Based on this idea, the basic version of our new algorithm operates in iterations. Within each iteration the algorithm makes multiple depth-first searches with a fixed expansion budget and varying $f$-cost limits. An iteration ends once the algorithm finds the optimal solution and expands all paths with $f$-cost less than $\coststar$, or once it can prove that the present expansion budget is insufficient to find the optimal solution. At the end of the iteration the expansion budget is doubled.

In what follows we abstract the search procedure into a query function that accepts as input an $f$-cost limit and an expansion budget and returns an interval that contains the smallest $f$-cost limit for which the budget is insufficient or throws an exception with an optimal solution if the $f$-cost limit is at least $C^*$ and the expansion budget is larger or equal to the number of nodes with $f$-cost less than the limit.

\paragraph{Formal model.}
We consider an increasing list $A$ of real numbers $v \ge 1$, possibly with repetition.
Define a function $n : [1,\infty) \to \Nonnegints$ by $n(\cost) = |\{v \in A : v \leq \cost\}|$, where multiple occurrences are counted separately. 
Next, let $\coststar \in A$
and $\nstar = n(\coststar)$.
In our application to tree search, $A$ is the list of all node $f$ values, including duplicates, and $n(\cost)$ is the number of paths in the search tree for which the $f$ values is at most $\cost$ and $\coststar$ is the cost of the optimal solution. We can require that all $f$ values are at least 1 with no loss of generality: if $h(\sinit) < 1$ we introduce an artificial new initial state with heuristic value 1 and an edge of cost $1 - h(\sinit)$ from the new state to $\sinit$. This shifts all path costs by at most 1, so if $C'$ is the original optimal solution cost, we have $C^* \le C' + 1$.

\paragraph{Query functions.}
Let $\cost_{\crit}(b) = \min\{v \in A : n(v) > b\}$ be
the smallest value in $A$ for which expansion budget $b$ is insufficient.
We define three functions, $\querylim$, $\queryint$ and $\queryext$, all accepting as input an $f$-cost limit $C$ and expansion budget $b$. A call to any of the functions makes at most $\min\{b, n(C)\}$ expansions and throws an exception with an optimal solution if $\nstar \leq n(C) \leq b$. Otherwise all three functions return an interval containing $\cost_{\crit}(b)$ on which we make different assumptions as described next.
The abstract objective is to make a query that finds an optimal solution using as few expansions as possible.

\paragraph{Limited feedback.}
In the limited feedback model the query function returns an interval that only provides information about whether or not the expansion budget $b$ was smaller or larger than $n(C)$:
\begin{multline*}
    \!\!\!\!\!\querylim(\cost, b) = 
    \begin{cases}
        [\cost,\infty] & \text{if } \cost < \cost_{\crit}(b)
        \quad\exptext{\scriptsize budget sufficient}\,, \\
        [1,\cost] & \text{if } \cost \geq \cost_{\crit}(b)
        \quad\exptext{\scriptsize budget exceeded}\,.
    \end{cases}
\end{multline*}

\paragraph{Integer feedback.}
In many practical problems, the list $A$ only contains integers. In this case we consider the feedback model:
\begin{align*}
    \queryint(\cost, b) =
    \begin{cases}
        [\floor{\cost}+1,\infty] &  \text{if } \cost < \cost_{\crit}(b)\,, \\
        [1,\cost] & \text{if } \cost \geq \cost_{\crit}(b)\,.
    \end{cases}
\end{align*}
The discrete nature of the returned interval means that if $\cost_{\crit}(b) \in [C, C+1]$, then 
\begin{align*}
    \queryint(C, b) \cap \queryint(C+1,b) = \{\cost_{\crit}(b)\}\,,
\end{align*}
In other words, the exact value of $\cost_{\crit}(b)$ can be identified by making queries on either side of an interval of unit width containing it. By contrast, $\querylim$ cannot be used to identify $\cost_{\crit}(b)$ exactly.

\paragraph{Extended feedback.}
For heuristic search problems the interval returned by the query function can be refined more precisely by using the smallest observed $f$-cost in the fringe and largest $f$-cost of an expanded path. Define 
   $ \tceil{\cost} = \min\{v \in A : v > \cost\}$ and 
   $\tfloor{\cost} = \max\{v \in A : v \leq \cost\}$.
When $\cost \in A$ we define $\delta(\cost) 
    = \tceil{C} - \tfloor{\cost}$ and
\begin{align*}
    \delta_{\min} &= \min\{\delta(\cost) : \cost \leq C^*,\, C \in A\}\,. 
\end{align*}
These concepts are illustrated in \cref{fig:ill}.
In the extended feedback model, when the expansion budget is sufficient the response of the query is the interval $[\tceil{\cost}, \infty]$. Otherwise the query returns an interval $[1, v]$ where $v$ is any value in $A \cap [\cost_{\crit}(b), \cost]$
(for example $v = \tfloor{\cost})$:
\begin{multline*}
    \queryext(\cost, b) = \\
    \begin{cases}
        [\tceil{\cost}, \infty] &  \text{if } \cost < \cost_{\crit}(b)\,, \\
        [1, v], \text{with }v \in [\cost_{\crit}(b), \cost] \cap A
        & \text{if } \cost \geq \cost_{\crit}(b)\,.
    \end{cases}
\end{multline*}
In tree search the value of $v$ when $C \geq \cost_{\crit}(b)$ is the largest $f$-cost over paths expanded by the search, which may depend on the expansion order.
As for integer feedback, the information provided by extended feedback
allows the algorithm to prove that an expansion budget is insufficient to find a solution. 
\begin{figure}[htb!]
\centering
\begin{tikzpicture}[scale=0.9,yscale=0.5]

 \draw[axis] (0,0) -- (8.0,0);
 \node[anchor=west] at (8.0,0) {$\cost$};
 \draw[axis] (0,0) -- (0,4.5);
 \node[anchor=south] at (0,4.5) {$n$};

 \draw[dotted] (0,1.5) -- (8.0,1.5); 
 \node[anchor=east] at (-0.1,1.5) {$b\!=\!3$};
 \draw[dotted] (0,0.5) -- (0.5,0.5); 
 \node[anchor=east] at (-0.1,0.5) {$1$};

 \draw[line] (0.5,0.5) -- (1.5,0.5); \draw[end] (0.5,0.5) circle (1pt);
\draw[line] (1.5,1.0) -- (2.5,1.0); \draw[end] (1.5,1.0) circle (1pt);
\draw[line] (2.5,2.0) -- (4.0,2.0); \draw[end] (2.5,2.0) circle (1pt);
\draw[line] (4.0,2.5) -- (6.0,2.5); \draw[end] (4.0,2.5) circle (1pt);
\draw[line] (6.0,3.0) -- (6.5,3.0); \draw[end] (6.0,3.0) circle (1pt);\draw [decorate,decoration={brace,amplitude=3pt}]
(6.0,3.1) -- (6.5,3.1);
\node[anchor=south] at (6.25,3.2) {$\delta_{\min}$};
\draw[line] (6.5,4.0) -- (7.5,4.0); \draw[end] (6.5,4.0) circle (1pt);

 \draw[line,-latex] (7.5,4.5) -- (8.0,4.5);\draw[end] (7.5,4.5) circle (1pt);

 \draw[dotted] (2.5,0) -- (2.5,2.0);
 \node[anchor=south] at (2.5,-2.1) {$\cost_{\crit}(b)$};
 \draw[dotted] (4.0,0) -- (4.0,2.5);
 \node[anchor=south] at (4.0,-2.1) {$\tfloor{\cost}$};
 \draw[dotted] (5.1,0) -- (5.1,2.5);
 \node[anchor=south] at (5.1,-2.0) {$\cost$};
 \draw[dotted] (6.0,0) -- (6.0,3.0);
 \node[anchor=south] at (6.2,-2.1) {$\tceil{\cost}$};
 \draw[dotted] (7.5,0) -- (7.5,4.5);
 \node[anchor=south] at (7.5,-2.0) {$\coststar$};

 \node[xa] at (0.5,0) {1};
\node[xa,color=gray] at (1.0,0) {2};
\node[xa] at (1.5,0) {3};
\node[xa,color=gray] at (2.0,0) {4};
\node[xa] at (2.5,0) {5};
\node[xa,color=gray] at (3.0,0) {6};
\node[xa,color=gray] at (3.5,0) {7};
\node[xa] at (4.0,0) {8};
\node[xa,color=gray] at (4.5,0) {9};
\node[xa,color=gray] at (5.0,0) {10};
\node[xa,color=gray] at (5.5,0) {11};
\node[xa] at (6.0,0) {12};
\node[xa] at (6.5,0) {13};
\node[xa,color=gray] at (7.0,0) {14};
\node[xa] at (7.5,0) {15};

 \draw [decorate,decoration={brace,amplitude=3pt}]
 (4.0,2.6) -- (6.0,2.6);
 \node[anchor=south] at (5.0,2.7) {$\fgap(\cost)$};

\end{tikzpicture}
\caption{The function $n(\cdot)$, generated by
    $b=3$, $\cost = 10.2$, $\coststar=15$ and
  $A = [1,3,5,5,8,12,13,13,15]$. 
  $\cost_{\crit}(b) = 5$ is the smallest value in $A$ for which there are more than $b=3$ values of at most 5.
  The largest value at most $C$ in $A$ is $\tfloor{\cost} = 8$,
  and $\tceil{\cost} = 12$ is the next value in $A$. We also have $\delta_{\min} = 1$ and $\delta(C) = 4$. 
  Example queries include: $\querylim(C = 4, b = 3) = [4,\infty]$ and $\queryext(C = 4, b = 3) = [5, \infty]$ and $\queryext(C = 9, b = 3) = [1, 8]$.
}\label{fig:ill}
\end{figure}
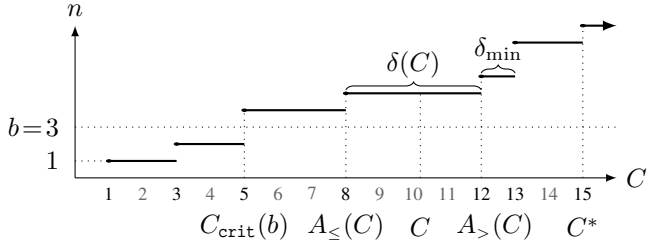

\paragraph{Summary of results.}
In the following sections we describe algorithms for all query models for which the number of node expansions is at most a logarithmic factor more than $\nstar$. The limited feedback model is the most challenging
and is detailed last, while the extended feedback model provides the cleanest illustration of our ideas.
The logarithmic factor depends on $\coststar$ and $\fgapmin$ or $\fgap(\coststar)$.
The theorems are summarized in \cref{tab:theorems}, with precise statements given in the relevant sections.

\begin{table}[ht]
\renewcommand{\arraystretch}{1.6}
\centering
\scalebox{0.95}{
\begin{tabular}{|ll|}
\hline
Limited feedback & $O(Z \log Z), Z = \nstar \log\left(\frac{\coststar}{\delta(\coststar)}\right)$ \\[0.2cm]
Extended feedback & $O\left(\nstar \log\left(\frac{\coststar}{\fgapmin}\right)\right)$ \\[0.2cm] 
Integer feedback & $O\left(\nstar \log(\coststar)\right)$ \\ \hline
\end{tabular}
}
\caption{Number of expansions in the worst case of our algorithms for the different types of feedback.}
\label{tab:theorems}
\end{table}

\paragraph{Overview.} In the next section we implement $\queryext$ for tree search and for graph search (\cref{sec:reductions}). 
We then introduce a variant of exponential search that uses the query function to find the `critical' cost for a given budget (\cref{sec:exp_search}).
Our main algorithm (\ibbs{}) uses the exponential search with a growing budget an optimal solution is found (\cref{sec:ibex}).
The \dovibbs{} algorithm is then provided  to deal with the more general limited feedback setting (\cref{sec:ubs,sec:dovibex}).


\section{Reductions}\label{sec:reductions}
We now explain how to reduce tree search and graph search to the abstract framework
and implement $\queryext$ for these domains.
These query functions will be used in the next sections as part of the main algorithms.
Recall that the number of expansions performed by $\query(\cost, \infty)$ must be at most $n(\cost)$ and $n(\cdot)$ is non-decreasing and that $\nstar = n(\coststar)$.

The list $A$ is composed of the $f$-costs of the nodes encountered
during the search.\todo{T: Something is fishy here. Encountered during what search? This should be the set of all $f$-costs of all paths I guess.} Recall from our problem definition that each node corresponds to a path $\pi$, and that expanding a node corresponds to expanding a state $s = \state(\pi)$.
For a search cost-bounded by $\cost$, let the fringe be all generated nodes with $f(\pi) > \cost$. Also, let the set of visited nodes be the generated nodes such that $f(\pi) \leq \cost$.

\paragraph{Tree search.}
For tree search we implement $\queryext$ in \cref{alg:tree_query}, which is a variant of depth-first search with an $f$-cost limit $\cost$ and expansion budget $b$. The algorithm terminates before exceeding the expansion budget and tracks the smallest $f$-cost observed in the fringe and the largest $f$-cost of any visited node, which are used to implement the extended feedback model.
Thus, for an $f$-cost bound $\cost$,
$n(\cost) = |\{\pi : \max_{s'\in\pi} f(s') \leq \cost\}|$.
When an optimal solution is found, the algorithm throws an exception, which is expected to be caught and handled by the user.
The function $\querylim$ could be implement like $\queryext$ without needing to track
\code{min_fringe} and \code{max_expanded}, but that would throw away valuable information.

\begin{algorithm}[htb!]
\begin{lstlisting}
def $\queryext$($\cost$, budget):
  data.min_fringe = $\infty$ # will be $\lstmathcomment{> C}$
  data.max_visited = 0 # will be $\lstmathcomment{\leq C}$
  data.expanded = 0 # number of expansions
  data.best_path = none # $\lstmathcomment{f(}$none$\lstmathcomment{) = \infty}$
  try:
    DFS($\cost$, budget, $\{\sinit\}$, data)
  catch "budget exceeded":
    return $[1, \texttt{data.max\_visited}]$
  if data.best_path $\neq$ none:  # solution found
    throw data # to be dealt with by the user
  return $[\texttt{data.min\_fringe}, \infty]$
# data: return info, passed by reference
# $\lstmathcomment{\pi}$: path
def DFS($\cost$, budget, $\pi$, data):
  if $f(\pi)$ > $\cost$: 
    data.min_fringe = min(data.min_fringe, $f(\pi)$) 
    return 
  data.max_visited = max(data.max_visited, $f(\pi)$)
  if $f(\pi) \geq f($data.best_path$)$:
    return # branch and bound
  if is_goal($\state(\pi)$):
    data.best_path = $\pi$
    # Here we could throw the solution if its
    # cost is equal to a known lower bound
    return
  if data.expanded == budget:
    throw "budget exceeded"
  data.expanded++
  for $s' \in \succf(\state(\pi))$:
    DFS($\cost$, budget, $\pi+\{s'\}$, data)
\end{lstlisting}

\caption{Query with extended feedback for tree search}
\label{alg:tree_query}
\end{algorithm}

\paragraph{Graph search.}
In graph search, $\queryext$ is implemented in a similar way as \cref{alg:tree_query},
but \code{DFS} is replaced with
\ifsup
\code{graph_search}, which appears in \cref{alg:graph_search} (\cref{sec:graph_search}).
\fi
The \code{graph_search} function is equivalent to BFIDA*~\cite{zhou2004structured} with the breadth-first search replaced with  Uniform-Cost Search (UCS)~\cite{russell2009aima,felner2011ucs},
using an $f$-cost limit $\cost$ on the generated nodes and tracking the maximum $f$-cost 
among visited states and the minimum $f$-cost in the fringe.
As in UCS, states are processed in increasing $g$-cost order.
Since $g$ is non-decreasing,
states are not expanded more than once in each
\ifsup call to \cref{alg:graph_search}.
\fi
Therefore the number of expansions is at most
$n(\cost) = |\{s: \min_{\pi:\state(\pi)=s}\max_{s'\in\pi}f(s') \leq \cost\}|$ and the number of expansions made by $\query(\coststar, \infty)$ is at most $n(\coststar) = \nstargraph$ as required.

\section{Exponential Search}\label{sec:exp_search}

\begin{algorithm}[htb!]
\begin{lstlisting}[escapeinside={(*}{*)}]
def exp_search(start, b, query):
  low = start
  high = $\infty$
  loop:
    if high == $\infty$:
      C = 2$\times$low  # exponential phase (*\label{line:bracket}*)
    else:
      C = (low + high) / 2 # binary phase
    [low, high] = [low, high] $\cap\,$query$(C, b)$
  until low == high  
  return low
\end{lstlisting}
\caption{Exponential search with budgeted queries}
\label{alg:es}
\end{algorithm}

With the reductions out of the way, we now introduce a budgeted variant of 
exponential search \cite{bentley1976expsearch}, which is closely related to the bracketed bisection method \cite[\S9]{FPT92}.

\cref{alg:es} accepts as input a budget $b$, an initial cost limit $\texttt{start} \leq \cost_{\crit}(b)$ and a function 
$\query \in \{\queryext, \queryint, \querylim\}$.
The algorithm starts by setting $\texttt{low} = \texttt{start}$ and initiates an exponential phase where \texttt{low} is repeatedly doubled until $\query(2 \times \texttt{low}, b)$ has insufficient budget.
The algorithm then sets $\texttt{high} = 2 \times \texttt{low}$ and performs a binary search on the interval $[\texttt{low}, \texttt{high}]$ until $\texttt{low} = \texttt{high}$. 
See \cref{fig:rv} for an illustration.


The discrete structure in the integer and extended feedback models ensures that the algorithm halts after at most logarithmically many queries and returns $\cost_{\crit}(b)$. In the limited feedback model the algorithm generally does not halt, but will make a terminating query
if $b \geq \nstar$. These properties are summarized in the next two propositions, which use the following definition:
\begin{align*}
    \nbisec(\eps, x, \Delta) &=
    1+\ceil{\log_2\left(\frac{x}{\eps}\right)}_{\geq 1} + \floor{\log_2\left(\frac{x}{\Delta}\right)}_{\geq 0}\,.
\end{align*}
This is an upper bound on the number of calls to query  needed when starting at $\texttt{start} = \eps$, finding a upper bound $\texttt{high} \geq x$
and then reducing the interval $[\texttt{low},\texttt{high}]$ to a size at most $\Delta$ (leading to a query within that interval).
Recall that making a query with cost limit $C$ and expansion budget $b$ will find an optimal solution  if $\nstar \leq n(\cost) \leq b$.

\begin{proposition}\label{prop:es-game-end}
Suppose $b \geq \nstar$. Then for any feedback model \cref{alg:es} makes a query that terminates the interaction
after at most
$\nbisec(\texttt{start},\ \ \coststar,\ \  \cost_{\crit}(b) - \coststar)$ calls to query.
\end{proposition}

\begin{proposition}\label{prop:es-nobudget}
Suppose $b < \nstar$. Then, for the extended feedback model, \cref{alg:es} returns $\cost_{\crit}(b)$ with at most
$\nbisec(\texttt{start},\ \ \cost_{\crit}(b),\ \  \fgapmin)$ queries.
\end{proposition}

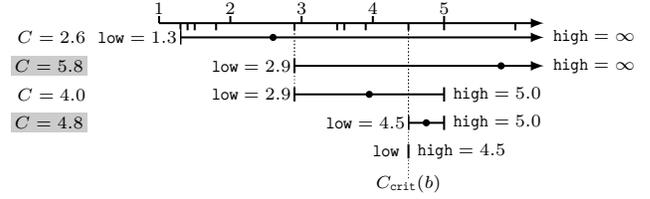
\begin{figure}[htb!]
    \centering
    \resizebox{\columnwidth}{!}{
    \begin{tikzpicture}[thick,font=\scriptsize]
    \draw[densely dotted,thin] (4.500000, 0.000000) -- (4.500000, -2.200000);
\draw[] (6.000000, 0.000000) -- (6.000000, -0.100000);
\draw[] (5.000000, 0.000000) -- (5.000000, -0.100000);
\draw[] (4.500000, 0.000000) -- (4.500000, -0.100000);
\draw[] (3.900000, 0.000000) -- (3.900000, -0.100000);
\draw[] (3.600000, 0.000000) -- (3.600000, -0.100000);
\draw[] (3.500000, 0.000000) -- (3.500000, -0.100000);
\draw[] (2.900000, 0.000000) -- (2.900000, -0.100000);
\draw[] (1.300000, 0.000000) -- (1.300000, -0.100000);
\draw[] (1.800000, 0.000000) -- (1.800000, -0.100000);
\draw[] (1.500000, 0.000000) -- (1.500000, -0.100000);
\draw[] (1.400000, 0.000000) -- (1.400000, -0.100000);
\draw[] (1.300000, -0.100000) -- (1.300000, -0.300000);
\node[anchor=east,fill=white,inner sep=1pt] at (1.300000, -0.200000) { $\texttt{low} = 1.3$ };
\draw[-latex] (1.300000, -0.200000) -- (6.400000, -0.200000);
\node[anchor=west] at (6.400000, -0.200000) { $\texttt{high} = \infty$ };
\draw[fill=black,draw=none] (2.600000, -0.200000) circle (0.050000);
\node[anchor=east,inner sep=0.5pt] at (0.000000, -0.200000) { $C = 2.6$ };
\draw[densely dotted, thin] (2.900000, -0.500000) -- (2.900000, -0.100000);
\draw[] (2.900000, -0.500000) -- (2.900000, -0.700000);
\node[anchor=east,fill=white,inner sep=1pt] at (2.900000, -0.600000) { $\texttt{low} = 2.9$ };
\draw[-latex] (2.900000, -0.600000) -- (6.400000, -0.600000);
\node[anchor=west] at (6.400000, -0.600000) { $\texttt{high} = \infty$ };
\draw[fill=black,draw=none] (5.800000, -0.600000) circle (0.050000);
\node[anchor=east,fill=black!20!white,inner sep=1.5pt] at (0.000000, -0.600000) { $C = 5.8$ };
\draw[densely dotted, thin] (2.900000, -0.900000) -- (2.900000, -0.100000);
\draw[] (2.900000, -0.900000) -- (2.900000, -1.100000);
\node[anchor=east,fill=white,inner sep=1pt] at (2.900000, -1.000000) { $\texttt{low} = 2.9$ };
\draw[] (5.000000, -0.900000) -- (5.000000, -1.100000);
\node[anchor=west] at (5.000000, -1.000000) { $\texttt{high} = 5.0$ };
\draw[] (2.900000, -1.000000) -- (5.000000, -1.000000);
\draw[fill=black,draw=none] (3.950000, -1.000000) circle (0.050000);
\node[anchor=east,inner sep=0.5pt] at (0.000000, -1.000000) { $C = 4.0$ };
\draw[densely dotted, thin] (4.500000, -1.300000) -- (4.500000, -0.100000);
\draw[] (4.500000, -1.300000) -- (4.500000, -1.500000);
\node[anchor=east,fill=white,inner sep=1pt] at (4.500000, -1.400000) { $\texttt{low} = 4.5$ };
\draw[] (5.000000, -1.300000) -- (5.000000, -1.500000);
\node[anchor=west] at (5.000000, -1.400000) { $\texttt{high} = 5.0$ };
\draw[] (4.500000, -1.400000) -- (5.000000, -1.400000);
\draw[fill=black,draw=none] (4.750000, -1.400000) circle (0.050000);
\node[anchor=east,fill=black!20!white,inner sep=1.5pt] at (0.000000, -1.400000) { $C = 4.8$ };
\node[anchor=east,fill=white] at (4.500000, -1.800000) { $\texttt{low}$ };
\node[anchor=west] at (4.500000, -1.800000) { $\texttt{high} = 4.5$ };
\draw[] (4.500000, -1.700000) -- (4.500000, -1.900000);
\draw[-latex] (1.000000, 0.000000) -- (6.400000, 0.000000);
\node[anchor=north] at (4.500000, -2.000000) { $\scriptsize C_{{\texttt{{crit}}}}(b)$ };
\draw[] (1.000000, 0.000000) -- (1.000000, 0.100000);
\node[anchor=south] at (1.000000, 0.000000) { $1$ };
\draw[] (2.000000, 0.000000) -- (2.000000, 0.100000);
\node[anchor=south] at (2.000000, 0.000000) { $2$ };
\draw[] (3.000000, 0.000000) -- (3.000000, 0.100000);
\node[anchor=south] at (3.000000, 0.000000) { $3$ };
\draw[] (4.000000, 0.000000) -- (4.000000, 0.100000);
\node[anchor=south] at (4.000000, 0.000000) { $4$ };
\draw[] (5.000000, 0.000000) -- (5.000000, 0.100000);
\node[anchor=south] at (5.000000, 0.000000) { $5$ };
    \end{tikzpicture}
    }
    \caption{The four queries made by \bisection{} with the extended feedback model and $\texttt{start} = 1.3$ and budget $\texttt{b} = 7$. The ticks below the $x$-axis indicate the elements of $A = [1.4, 1.5, 1.8, 2.3, 2.9, 3.5,3.6,3.9,4.5,5,6]$. The small circles are the values of $\cost$ in each query. In the first call to $\queryext$ the budget was sufficient and so $\texttt{low}$ is set to $A_{>}(C) = 2.9$, which is doubled to produce the next query. In the second call to query, $C = 5.8$ leads to an insufficient budget and then $\texttt{high}$ is set to $5.0$. In the third query $\costlow$ is increased to 4.5, which is $\cost_{\crit}(b)$. In the fourth query the budget is insufficient and the algorithm halts.}
    \label{fig:rv}
\end{figure}
 
\section{\ibbslong}\label{sec:ibex}

The \ibbslong{} (\ibbs) algorithm uses the extended query model, which is available in our applications to tree and graph search. 
\cref{alg:rv} initializes a lower bound on the optimal cost with the lowest value in the array $C_1 = \min A$; we now denote this quantity by $\cost_{\min}$. It subsequently operates in iterations $k \in \Nonnegints$.
In iteration $k$ it sets the budget to $b_k = 2^k$ and 
calls $C_{k+1} = \bisection(C_k, b_k,\queryext)$
to obtain a better lower bound.

\begin{algorithm}[b!]
\begin{lstlisting}[escapeinside={(*}{*)}]
def IBEX():  # simple version
  $\cost_1$ = $\cost_{\min}$
  for k = 1,2,...
    $b_k$ = $2^k$
    $\cost_{k+1}$ = exp_search($\cost_k$, $b_k$, $\queryext$)
\end{lstlisting}
\caption{Iterative Budgeted Exponential Search}
\label{alg:rv}
\end{algorithm}

\begin{theorem}\label{thm:rv}
The number of expansions made by \cref{alg:rv} is at most
\begin{align*}
    4\nstar\nbisec({\cost_{\min}, \coststar, \fgapmin}) = 
    O\left(\nstar\log\left(\frac{\coststar}{\fgapmin} \right)\right)\,.
\end{align*}
\end{theorem}

\begin{proof}
Define $r_1 = \nbisec(\costmin, \coststar, \fgapmin)$.
Let $k_*=\ceil{\log_2 \nstar}$ be the first iteration $k$ for which $b_k \geq \nstar$.
\cref{prop:es-nobudget} shows that
for iterations $k < k_*$ the number of queries performed in the call to \bisection{} is at most
$\nbisec(\cost_k, \cost_{\crit}(b_k), \fgapmin)
\leq r_1$.
\cref{prop:es-game-end} shows that the game ends during iteration $k_*$ after a number of queries bounded by $\nbisec(\cost_{k_*}, \coststar, \cost_{\crit}(b_{k_*})- \coststar) \leq r_1$.
Since each call to query with budget $b_k$ expands at most $b_k=2^k$ nodes,
the total number of expansions is bounded by
    $\sum_{k=1}^{k_*} 2^k r_1 \leq 2^{k_*+1} r_1 \leq 2^{2+\log_2 \nstar}r_1 
    = 4\nstar r_1$.
\end{proof}

\begin{remark}\label{rem:rv}
\cref{alg:rv} also works when 
$\queryext$ is replaced by $\queryint$,
and now $\fgapmin=1$.
\end{remark}

When used for graph search, we call the resulting \ibbs{} variant
\ibbsgslong{} (\ibbsgs), 
and for tree search \ibbstslong{} (\ibbsts).

\begin{remark}
Observe that if \code{DFS} or \code{graph_search} are called with $\texttt{budget} \geq n(\cost) \geq \nstar$,
it throws an optimal solution.
Since the state space is finite, both \ibbsts{} and \ibbsgs{} return an optimal solution if one exists. If no solution exists, \ibbsgs{} will exhaust the graph and return ``no solution'', but \ibbsts{} may run forever unless additional duplicate detection is performed.
\end{remark}

\section{Uniform Budgeted Scheduler}\label{sec:ubs}
\newcommand{\seg}{r} 
\newcommand{\segalt}{m}

While IBEX handles the integer and extended feedback models, it cannot handle the limited feedback model. This is handled by a new algorithm, presented in the next section, using a finely balanced dovetailing idea that we call the \ubslong{} (\ubs, see \cref{alg:ubs}) and takes inspiration from \citeauthor{luby1993speedup} (\citeyear{luby1993speedup}) speedup algorithm.
UBS runs a growing and unbounded number of programs in a dovetailing fashion,
for varying segments of \emph{steps}.
The notion of \emph{step} is to be defined by the user; in heuristic search we take it to be a single node expansion.
During one segment, the selected program
can make arbitrary computations but must use no more steps than its current budget.
Program $k$ halts when it reaches exactly $\tau_k$ steps, which may be infinite. 
UBS maintains a priority queue of pairs (program index $k$, segment number $r$), initialized with $(1, 1)$ and ordered by a function $T : \Naturals \times \Nonnegints \to \Nonnegints$ with $T(k, r) < T(k, r+1)$ and $T(k,r) \leq T(k+1,r)$ for all $(k,r)$ and $T(k, 0) = 0$ for all $k$.
In each iteration \ubs{} removes the $T$-minimal element $(k, r)$ from the front of the queue and calls \code{run_prog}$(k, b)$
with $b=T(k, r) - T(k, r-1)$,
which means that program $k$ is being run for its $r$th segment with a budget of $b$ steps, leading for program $k$ to a total of at most $T(k, r)$ steps over the $r$ segments.
If this does not cause program $k$ to halt, then $(k, r+1)$ is added to the priority queue. Finally, if $r = 1$, then $(k+1,1)$ is added too.
The function \code{run_prog} is defined by the user and may store and restore the state of program $k$ as well as allow access to a shared memory space.

\begin{algorithm}[tb!]
\begin{lstlisting}
def UBS($T$, run_prog):
  q = make_priority_queue($T$)
  q.insert((1, 1))  # k=1, r=1
  while not q.empty():
    # Remove prog of minimum T cost
    $(k, r)$ = q.extract_min()
    budget = $T(k,r) - T(k,r-1)$
    if run_prog($k$, budget) != "halted":
      q.insert(($k, r+1$))
    if $r$ == 1:
      q.insert(($k+1, 1$))
  return none
\end{lstlisting}
\caption{\ubslong}
\label{alg:ubs}
\end{algorithm}

The monotonicity assumptions on $T$ mean that UBS is essentially 
executing program/segment pairs $(k, r)$ according to a Uniform Cost Search
where a pair $(k, r)$ is a node in the search tree
with cost $T(k, r)$ (\cref{fig:ubs-tree}).
In this sense \ubs{} tries (asymptotically) to maintain a uniform amount of steps used among all non-halting programs.

\begin{figure}[b!]
\centering
\begin{tikzpicture}[font=\scriptsize,yscale=0.8]
\tikzstyle{n} = [inner sep=1pt]
\node[n] (11) at (1,0) {$(k = 1,r = 1)$};
\node[n] (12) at (3.8,0) {$(k = 1,r = 2)$};
\node[n] (13) at (6.6,0) {$(k = 1,r = 3)$};
\node[n] (14) at (8.4,0) {};
\node[n] (21) at (1,-0.8) {$(k = 2,r = 1)$};
\node[n] (22) at (3.8,-0.8) {$(k = 2,r = 2)$};
\node[n] (23) at (5.5,-0.8) {halts};
\node[n] (31) at (1,-1.6) {$(k = 3,r = 1)$};
\node[n] (32) at (3.8,-1.6) {$(k = 3,r = 2)$};
\node[n] (33) at (6.6,-1.6) {$(k = 3,r = 3)$};
\node[n] (34) at (8.4,-1.6) {};
\draw[-latex] (11) -- (12);
\draw[-latex] (12) -- (13);
\draw[-latex] (13) -- (14);
\draw[-latex] (11) -- (21);
\draw[-latex] (21) -- (22);
\draw[-latex] (21) -- (31);
\draw[-latex] (31) -- (32);
\draw[-latex] (32) -- (33);
\draw[-latex] (33) -- (34);
\draw[-latex] (31) -- (1,-2.15);
\end{tikzpicture}
\caption{An example tree used by \ubs. Program $k = 2$ halts after  $\tau_2 \leq T(2,2)$ steps.}\label{fig:ubs-tree}
\end{figure}
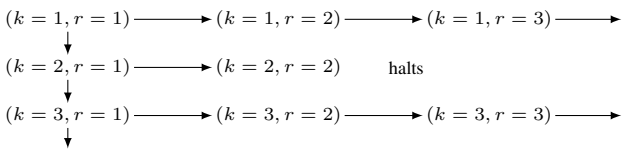

Let $T_{\ubs}(k, \seg)$ be the number of steps used by \ubs{} after executing $(k, \seg)$.
$T_j(k, \seg) = \max\{T(j, m) : T(j, m) \leq T(k, \seg), m \in \Nonnegints\}$ is an upper bound on the number of steps used by program $j$ when \ubs{}  executes $(k, \seg)$. 

\begin{theorem}\label{thm:ubs}
For any $(k, \seg)$ with $T(k, \seg-1) < \tau_k$,
\ifsup
\begin{align*}
T_{\ubs}(k, \seg) &\leq \sum_{j\in\Naturals} \min\{\tau_j, T_j(k,\seg)\} \,.
\end{align*}
\fi
\end{theorem}

\ifsup
\begin{corollary}\label{cor:ubs}
For any $(k, \seg)$ with $T(k, \seg-1) < \tau_k$,
\begin{align*}
T_{\ubs}(k,\seg) \leq T(k, \seg) \max\{j: T(j,1) \leq T(k, \seg)\} \,.
\end{align*}
\end{corollary}
\else
\fi

\begin{example}\label{example:ubs}
A good choice is $T(k, \seg) = \seg 2^k$.
Then \ifsup \cref{cor:ubs} \fi  implies that $T_{\ubs}(k,\seg) \leq \seg 2^k \floor{k + \log_2(\seg)}$.
\end{example}

\section{\dovibbs: Limited Feedback}\label{sec:dovibex}

\dovibbs{} uses \ubs{} to dovetail multiple instances of exponential search (see \cref{alg:zzzprog}).
The algorithm can use any of the three queries, while still exploiting the additional information provided in the extended and integer feedback models.
Following \cref{example:ubs}, we use $T(k, r) = r 2^k$ so
that $T(k, r) - T(k, r-1) = 2^k$. 
Program $k$ executes one iteration of the loop of exponential search with budget $b_k = 2^k$; if the budget is not entirely used during a segment, the segment ends early.
In the limited feedback model, exponential search may continue halving the interval $[\costlow, \costhigh]$ indefinitely and never halt.
The scheduler solves this issue: by interleaving multiple instantiations of exponential search with increasing budgets, the total time can be bounded as a function 
of the time required by the first program $k$ that finds a solution.

\begin{theorem}\label{thm:zzz}
For $\query \in \{\queryext, \queryint, \querylim\}$, 
the number of expansions made by \cref{alg:zzzprog} is at most
\begin{align*}
    \quad\quad Z\floor{\log_2 Z}\ ,\quad\quad Z = 2\nstar\nbisec(\cost_{\min}, \coststar,\fgap(\coststar))
\end{align*}
and also at most $O(Z'\log Z')$ with
\begin{align}
Z' = \min_{b\geq \nstar} b\,\nbisec(\cost_{\min}, \coststar, \cost_{\crit}(b)-\coststar)\,. \label{eq:z3_minb}
\end{align}
\end{theorem}

The dependency on $\fgap(\coststar)$ is gentler than the dependency on $\fgapmin$ in \cref{thm:rv}. 
It means that the behaviour of \dovibbs{} only depends on the structure of the search space at the solution rather than on the worst case of all iterations.

\begin{remark}
Sometimes there exists a $b > \nstar$ for which
\begin{align*}
b\nbisec(\cost_{\min}, \coststar, \cost_{\crit}(b) - \coststar)
\ll
\nstar\nbisec(\cost_{\min}, \coststar, \fgapmin)\,.
\end{align*}
In these cases the combination of \ubs{} and exponential search can improve on \cref{alg:rv}. 
Furthermore when using extended feedback, programs $k < k_*$ will now halt if they can prove they cannot find a solution.
This allows us to provide the following complementary bound,
which is only an additive $2\nstar \seg_2\floor{\log_2 \seg_2}$ term away from \cref{thm:rv}.
\end{remark}

\begin{theorem}\label{thm:zzz2}
When $\query = \queryext$, the number of expansions made by \cref{alg:zzzprog} is at most 
$2\nstar (\seg_1 + \seg_2(1 + \floor{\log_2 \seg_2})$
with $\seg_1 = \nbisec(\cost_{\min}, \coststar, \fgapmin)$,
    $\seg_2 = \nbisec(\cost_{\min}, \coststar, \fgap(\coststar))$.
\end{theorem}

\begin{algorithm}[htb!]
\begin{lstlisting} 
def run_prog(k, b):
  $[\costlow, \costhigh]$ = get_state(k, default = $[\cost_{\min}, \infty]$)
  if $\cost_{\high} = \infty$:  # exponential phase
    $\cost$ = $2\times\cost_{\low}$
  else:  # binary phase
    $\cost$ = $(\cost_{\low} + \cost_{\high})/2$
  $[\costlow, \costhigh] = [\costlow, \costhigh] \cap\,$query$(\cost, b)$
  if $\costlow$ >= $\costhigh$: return "halted"
  store_state(k, $[\costhigh, \costlow]$)
def DovIBEX():
  UBS($(k, r) \mapsto r 2^k$, run_prog)
\end{lstlisting}
\caption{\dovibbslong}\label{alg:zzzprog}
\end{algorithm}
 
\section{Enhancements}\label{sec:enhancements}





We now describe three enhancements for \cref{alg:rv} (see  \cref{alg:rv_enhanced}). The enhancements use a modified query function called $\queryextp$ that returns the same interval as $\queryext$ and the number of expansions, which is  $n_{\used} = \min\{b, n(\cost)\}$. For notational simplicity we write 
$[x,y], n_{\used} = \queryextp(\cost, b)$.

First, in each iteration of the enhanced IBEX a query is performed with an infinite expansion budget
and minimum cost limit $\costlow$
(\cref{alg:en:inf}). If the resulting number of expansions is at least $2b$, where $b$ is the current budget, then \ibbs{} updates $\costlow$, skips the exponential search and moves directly to the next iteration.
If furthermore the DFS algorithm is given the lower bound $\cost_{k}$
and throws a solution if its cost is $\cost_k$, then
this guarantees that \ibbs{} performs exactly like IDA* in domains in which the number of expansions grows by at least a factor 2 in each iteration.
If the queries with infinite budget (\cref{alg:en:inf}) do not help skipping iterations,
in the worst case they cost an additive $2\nstar$ expansions. 

The second enhancement is an early stopping condition for proceeding to the next iteration. \cref{alg:rv} terminates an iteration once it finds $\cost_{\crit}(b)$, which can be slow when $\fgap(\cost_{\crit}(b))$ is small.
\cref{alg:rv_enhanced} uses a budget window defined by $2b$ and  $\alpha b$, where $\alpha \geq 2$, so that whenever a query is made and the number of expansions is in the interval $[2b, \alpha b]$,
the algorithm moves on to the next iteration (line~\ref{alg:en:early_stopping}).
Hence iteration $k$ ends when a query is made within budget with a cost within $[\cost_{\crit}(2b), \cost_{\crit}(\alpha b)]$. 

The third enhancement is the option of using 
an additive variant of the exponential search algorithm,
which increases $\costlow$ by increments of $2^j$ at iteration $j$
during the exponential phase (line~\ref{alg:en:additive}).
This variant is based on the assumption that costs increase linearly
when the budget doubles, which often happens in heuristic search
if the search space grows exponentially with the depth.

These enhancements can also be applied to \dovibbs{} 
(see \ifsup{}\cref{alg:zzz_fast} in \cref{sec:zzz_fast}\fi).

\begin{algorithm}[htb!]
\begin{lstlisting}[escapeinside={(*}{*)}]
# $\textcolor{lstcommentcolor}{\alpha}$: factor on budget, which must be $\textcolor{lstcommentcolor}{\geq 2}$
# is_additive: True is using additive search
def IBEX_enhanced($\alpha$=8, is_additive):
  $\costlow$ = $\cost_{\min}$
  $b$ = 1
  for $k = 1,2,\ldots$:
    $[\costlow, \costhigh]$, $n_\text{used}$ = $\queryextp$($\costlow$, $\infty$) (*\label{alg:en:inf}*)
    if $n_{\text{used}} < 2b$:
      for $j=1, 2,\ldots$: (*\label{alg:en:regular_query}*)
        if $\costhigh$ == $\infty$:  # exponential phase
          if is_additive:
            $\cost = \costlow + 2^j$ (*\label{alg:en:additive}*)
          else:
            $\cost$ = $\costlow\times 2$
        else:  # binary phase
          $\cost$ = $(\costlow + \costhigh)/2$
        $[\costlow', \costhigh']$, $n_{\text{used}}$ = $\queryextp(\cost, \alpha \times b)$
        $[\costlow,\costhigh] = [\costlow,\costhigh] \cap [\costlow',\costhigh']$
        if ($\costhigh'$ == $\infty$ and $n_{\text{used}} \geq 2b$) or (*\label{alg:en:early_stopping}*)
            $\cost_{\low}$ == $\cost_{\high}$:
          break (*\label{alg:en:regular_query_end}*)
    $b$ = max(2$b$, $n_{\text{used}}$) (*\label{alg:en:max_budget}*)
\end{lstlisting}
\caption{Enhanced \ibbs}
\label{alg:rv_enhanced}
\end{algorithm}

\section{Experiments}
\newcommand{\deepening}{DFS}



\newif\iffull\fullfalse

\newcommand{\hlh}[1]{\iffull{#1}\else\textbf{#1}\fi} 

\begin{table*}[htb!]
    \centering
\resizebox{\textwidth}{!}{
    \rowcolors{2}{gray!10}{white}
    \begin{tabular}{lccrrrrrrrrrr}
\toprule
\rowcolor{white}
\multicolumn{3}{c|}{Algorithm}& \multicolumn{2}{c|}{15-Puzzle (unit)} & \multicolumn{2}{c|}{15-Puzzle (real)} & \multicolumn{2}{c|}{(12, 4)-Topspin} &\multicolumn{2}{c|}{Chain} &\multicolumn{2}{c}{Coconut} \\
\rowcolor{white}
&$\alpha$&\multicolumn{1}{c|}{add?} & 
Solved & \multicolumn{1}{c|}{Exp.} &
Solved & \multicolumn{1}{c|}{Exp.} &
Solved & \multicolumn{1}{c|}{Exp.$\times10^3$} &
Solved & \multicolumn{1}{c|}{Exp.$\times10^4$} &
Solved & Exp.$\times10^4$ \\
\midrule  
                 \ibbsts{}      &  2   &  y  &  100  &  \hlh{242.5}  &  97   &  3\,214.1     &  100  &  1\,521.9  &  100  &  302.0  &  100  &  72.9 \\
\iffull \phantom{\ibbsts{}}     &  4   &  y  &  100  &  \hlh{242.5}  &  100  &  1\,104.4     &  100  &  743.3     &  100  &  198.2  &  100  &  63.0 \\ \fi
        \phantom{\ibbsts{}}     &  8   &  y  &  100  &  \hlh{242.5}  &  100  &  \hlh{673.1}  &  100  &  597.0     &  100  &  198.2  &  100  &  84.7\\
\iffull \phantom{\ibbsts{}}     &  16  &  y  &  100  &  \hlh{242.5}  &  100  &  573.5        &  100  &  610.0     &  100  &  198.2  &  100  &  132.2\\\fi

        \phantom{\ibbsts{}}     &  2   &  n  &  100  &  \hlh{242.5}  &  97   &  3\,549.1     &  100  &  1\,600.0  &  100  &  111.8  &  100  &  \hlh{58.5}\\
\iffull \phantom{\ibbsts{}}     &  4   &  n  &  100  &  \hlh{242.5}  &  100  &  1\,568.5     &  100  &  817.5     &  100  &  26.7   &  100  &  41.7\\\fi
        \phantom{\ibbsts{}}     &  8   &  n  &  100  &  \hlh{242.5}  &  99   &  1\,320.3     &  100  &  614.6     &  100  &  26.7   &  100  &  86.8\\
\iffull \phantom{\ibbsts{}}     &  16  &  n  &  100  &  \hlh{242.5}  &  100  &  1\,548.1     &  100  &  607.5     &  100  &  26.7   &  100  &  135.1\\\fi

                  \dovibbsts{}  &  2   &  y  &  100  &  390.5        &  100  &  1\,087.7     &  100  &  1\,083.1  &  100  &  287.2  &  100  &  107.4 \\ 
\iffull \phantom{\dovibbsts{}}  &  4   &  y  &  100  &  322.0        &  100  &  928.6        &  100  &  723.9     &  100  &  175.5  &  100  &  291.9 \\\fi 
        \phantom{\dovibbsts{}}  &  8   &  y  &  100  &  322.0        &  100  &  767.4        &  100  &  606.1     &  100  &  125.9  &  100  &  1\,136.0 \\ 
\iffull \phantom{\dovibbsts{}}  &  16  &  y  &  100  &  320.1        &  100  &  695.1        &  100  &  456.3     &  100  &  120.7  &  100  &  16\,860.6 \\\fi 

        \phantom{\dovibbsts{}}  &  2   &  n  &  100  &  322.0        &  98   &  2\,355.2     &  100  &  1\,145.1  &  100  &  33.8   &  100  &  121.8 \\ 
\iffull \phantom{\dovibbsts{}}  &  4   &  n  &  100  &  322.0        &  99   &  2\,264.3     &  100  &  725.3     &  100  &  26.6   &  100  &  439.9 \\\fi 
        \phantom{\dovibbsts{}}  &  8   &  n  &  100  &  474.6        &  100  &  2\,432.2     &  100  &  590.5     &  100  &  24.9   &  100  &  2\,882.9 \\ 
\iffull \phantom{\dovibbsts{}}  &  16  &  n  &  100  &  1\,465.5     &  97   &  1\,933.0     &  100  &  637.0     &  100  &  25.8   &  100  &  13\,774.4 \\\fi 

           \multicolumn{3}{l}{EDA* $\gamma=2$}  &  99   &  5\,586.0     &  100  &  2\,882.3   &  100  &  807.5        &  100  &  \hlh{19.4}  &  3    &  $\geq 554\,249.0$ \\
 \iffull \multicolumn{3}{l}{EDA* $\gamma=1.1$}  &  100  &  1\,347.4     &  100  &  694.4      &  100  &  234.9        &  100  &  107.5       &  3    &  $\geq 564\,518.5$\\\fi
        \multicolumn{3}{l}{EDA* $\gamma=1.01$}  &  100  &  1\,023.8     &  100  &  742.0      &  100  &  730.2        &  100  &  990.2       &  10   &  $\geq 528\,137.9$\\
                                    IDA*\_CR&&  &  100  &  868.4        &  100  &  700.6      &  100  &  \hlh{346.0}  &  100  &  988.3       &  3    &  $\geq 516\,937.7$\\
                                       IDA*&&   &  100  &  \hlh{242.5}  &  57   &  62\,044.3  &  100  &  2\,727.5     &  100  &  162\,129.0  &  100  &  5\,484.2\\
\hline
$\nstartree^{\scriptstyle <}$&  &&  100 & 100.8   &  100 & 258.1  &  100  &  35.8   &  100  &  4.9  &  100  &  2.7\\
\bottomrule
    \end{tabular}
}
    \caption{Results on tree search domains.
    Each tasks has 100 instances.
    Expansions (Exp.) are averaged on solved tasks only (except Coconut), and times $10^6$ for the 15-puzzle.
    \ibbsts{} is the implementation for tree search of the \ibbs{} framework.
    add? is the \texttt{is\_additive} parameter.
    \todo{$\nstartree^{\scriptstyle <}$ is for $< \coststar$}
    }
    \label{tab:results}
\end{table*}

\begin{table}[htb!]
    \centering
\resizebox{\columnwidth}{!}{%
    \rowcolors{2}{gray!10}{white}
\begin{tabular}{lccrrrr}
        \toprule
        \rowcolor{white}
        \multicolumn{3}{c|}{Algorithm}& \multicolumn{1}{c|}{d=$100$}& \multicolumn{1}{c|}{d=$1\,000$}& \multicolumn{2}{c}{d=$10\,000$}\\
        \rowcolor{white}
        &$\alpha$& \multicolumn{1}{@{}c@{}|@{}}{add?\,} & \multicolumn{1}{r|}{Exp.} & \multicolumn{1}{r|}{Exp.} & Exp. & \multicolumn{1}{@{}r@{}}{Time\,(s)} \\
        \midrule
                  \ibbsgs{} & 2  & y  & 2\,592 & 35\,478  & 752\,392  & 0.4 \\
\iffull \phantom{\ibbsgs{}} & 4  & y  & 1\,279 & 22\,278 & 312\,500 & 0.2 \\\fi
        \phantom{\ibbsgs{}} & 8  & y  & 1\,276 & 22\,275 & 312\,497  & 0.2 \\
\iffull \phantom{\ibbsgs{}} & 16 & y  & 1\,309 & 25\,843 & 322\,662 & 0.2 \\\fi
        \phantom{\ibbsgs{}} & 2  & n  & 2\,429 & 26\,030 & 513\,573  & 0.4 \\
\iffull \phantom{\ibbsgs{}} & 4  & n  & 769 & 8\,551 & 82\,363  & 0.1 \\\fi
        \phantom{\ibbsgs{}} & 8  & n  & 513 & 8\,821 & 84\,434  & 0.1 \\
\iffull \phantom{\ibbsgs{}} & 16 & n  & 811 & 8\,160 &  35\,266 & 0.1 \\\fi
        \dovibbsgs{}           & 2  & y & 2\,195 & 31\,862 & 564\,720 & 0.1  \\
\iffull \phantom{\dovibbsgs{}} & 4  & y & 1\,621 & 26\,325 & 410\,131 & 0.1 \\\fi
        \phantom{\dovibbsgs{}} & 8  & y & 1\,495 & 15\,757 & 189\,883 & 0.1 \\
\iffull \phantom{\dovibbsgs{}} & 16 & y & 1\,046 & 11\,987 & 185\,239 & 0.1 \\\fi
        \phantom{\dovibbsgs{}} & 2  & n & 1\,547 & 12\,987 & 185\,500 & 0.1 \\
\iffull \phantom{\dovibbsgs{}} & 4  & n & 531    & 7\,665  & 103\,135 & 0.1 \\\fi
        \phantom{\dovibbsgs{}} & 8  & n & \hlh{449}    & \hlh{4\,017}  & \hlh{36\,093} & 0.1 \\
\iffull \phantom{\dovibbsgs{}} & 16 & n & 252    & 2\,789  & 33\,432 & 0.1  \\\fi
         A*/B/B'&&                            & 7\,652 & 751\,502  & 75\,015\,002  & 22.4 \\ 
\hline         
         $\nstargraph^{\scriptstyle <}$&& &   200 & 2\,000 & 20\,000 & 0.0  \\
        \bottomrule
    \end{tabular}
}
    \caption{Results for inconsistent heuristics in graph search. 
    \ibbsgs{} is the implementation of the IBEX framework for graph search.
    \todo{$\nstargraph^{\scriptstyle <}$}}
    \label{tab:inconsistent}
\end{table}

We test\footnote{
All these algorithms are implemented in C++ in the publicly available HOG2 repository, \url{https://github.com/nathansttt/hog2/}.}
\ibbs{} (\ibbsts{}, enhanced), \dovibbs{} (\dovibbsts{}, enhanced),
\ida{} \cite{korf1985ida}, \idacr{} \cite{sarkar1991reducing} and \eda{} \cite{sharon2014eda}.
\eda($\gamma$) is a variant of \ida{} designed for polynomial domains that repeatedly calls \deepening{} with unlimited budget and a cost threshold of $\gamma^k$ at iteration $k$.
In our experiments we take $\gamma \in \{2,1.01\}$.
\idacr{} behaves similarly, but adapts the next cost threshold by collecting the costs of the nodes
in the fringe into buckets and selecting the first cost that is likely to expand at least $b^k$ nodes in the next iteration.
Our implementation uses 50 buckets and sets $b=2$.
The number of nodes (states) of cost strictly below $\coststar$ is reported as $\nstartree^{\scriptstyle<}$ ($\nstargraph^{\scriptstyle<}$).

These algorithms are tested for tree search on the 15-Puzzle \cite{DM66}
with the Manhattan distance heuristic
 with unit costs and with varied edge costs of $1+1/(t+1)$ to move tile $t$,
on (12, 4)-TopSpin \cite{Lam89} with random action costs between 40 and 60 and the max of 3 4-tile pattern database heuristics, on long chains (branching factor of 1 and unit edge costs, solution depth in $[1..10\,000]$), and on a novel domain, which we explain next. 

In order to evaluate the robustness of the search algorithms, we introduce the Coconut problem, which is a domain with varied branching factor and small solution density.
The heuristic is 0 everywhere, except at the root where $h=1$.
At each node there are 3 `actions', $\{1, 2, 3\}$.
The solution path follows the same action (sampled uniformly in $[1..3]$) for $D$ steps, then it follows a random path sampled uniformly of length $q$, where
$D$ is sampled uniformly in $[1..10\,000]$ and
$q$ is sampled from a geometric distribution of parameter $1/4$.
The first action costs 1.
At depth less than $D$, taking the same action as at the parent node costs 1,
taking another action costs $2D$.
At depth larger than $D$, each action costs $1/10$.

\subsection{Results}
We use 100 instances for each problem domain,
a time limit of 4 hours for 15-Puzzle and TopSpin, 1 hour for the Coconut problem and no limit for the Chain problem. 
The results are shown in \cref{tab:results}.
We report results also for $\alpha=2$ to show the gain in efficiency when using a budget factor window of $[2, 8]$ instead of the narrower window $[2,2]$ (see \cref{sec:enhancements}).
\ibbs{} (\ibbsts) and \dovibbs{} (\dovibbsts) ($\alpha=8$) are the only robust algorithms across all domains while being competitive on all domains, whereas all other algorithms tested fail hard on at least one domain.
\ibbs{} (\ibbsts) has exactly the same behaviour as \ida{} when the number of expansions grows at least by a factor 2 at each call to $\query$ with infinite budget; see 15-Puzzle (unit).
Taking is\_additive=y helps on exponential domains, whereas is\_additive=n helps on polynomial domains, as expected.

To explain the behaviour of \idacr\ and \eda\ on the Coconut problem, consider a randomly chosen instance where $D=2\,690$ and $q=6$. The cost set by \eda{} ($\gamma=2$) in the last iteration is 4\,096, resulting in a search tree with approximately $3^{(4\,096-2\,690)/0.1}\approx 10^{6\,700}$ nodes. 
The same issue arises for \idacr.
\eda(1.01) performs only marginally better.
This is not a carefully selected example, and such behaviour occurs on almost all Coconut instances.

Finally, we evaluate our algorithms in graph search problems with inconsistent heuristics,
parameterizing \citeauthor{mero1984pathmax}'s (\citeyear{mero1984pathmax}) graph
by $d$ to have $2d+2$ states (see \ifsup{}\cref{fig:mero} in \cref{sec:mero}\fi).
All states have heuristic of 0 except each state $t_i$ which has heuristic $d+i-1$.
A*, B \cite{Mar77}, and B' \cite{mero1984pathmax} are all expected to perform $O(\nstargraph^2)$ expansions on this graph.
The results are in \cref{tab:inconsistent}.
While A* shows quadratic growth on the number of expansions, our algorithms exhibit near-linear performance as expected.

\section{Conclusion}

We have developed a new framework called \ibbs{} that combines exponential search with an increasing node expansion budget to resolve two long-standing problems in heuristic search. The resulting algorithms for tree and graph search improve existing guarantees on the number of expansions from $\Omega(\nstar^2)$ to $O(\nstar \log \coststar)$. Our algorithms are fast and practical. They significantly outperform existing baselines in known failure cases while being at least as good, if not better, on traditional domains; hence, for tree search we recommend using our algorithms instead of IDA*.
On graph search problems our algorithms outperform A*, B and B' when the heuristic
is inconsistent, and pay only a small $\log \coststar$ factor otherwise.

We also expect the IBEX framework
to be able to tackle re-expansions problems of other algorithms,
such as Best-First Levin Tree Search \cite{orseau2018single} and Weighted A* \cite{chen19socs}.

\comment{LL: I suggest removing the following paragraph. This is work yet to be done. LO: That's called an outreach ;)}
The IBEX framework and algorithms have potential applications beyond search
in domains that exhibit a dependency between a parameter and the amount of work (computation steps, energy, etc.) required to either succeed or fail.
Some of these applications may not be well suited to the extended feedback model, which further justifies the interest in the analysis of the limited feedback model.

\section*{Acknowledgements}
This research was enabled in part by a sabbatical grant from the University of Denver and by Compute Canada (www.computecanada.ca).
Many thanks to 
Csaba Szepesv\'ari,
Andr\'as Gy\"orgy,
J\'anos Kram\'ar,
Roshan Shariff, Ariel Felner, 
and the 
reviewers for their feedback.


\appendix

\bibliographystyle{named}
\bibliography{ijcai19}

\clearpage

\ifsup

\section{Graph Search Query}\label{sec:graph_search}
The function \code{graph_search} is given in \cref{alg:graph_search}.
The actual query function is like \cref{alg:tree_query}
where the call to \code{DFS} is replaced with a call to \code{graph_search}.
We use nodes instead of paths to stress that each element of the priority queue
uses constant memory size.
The fringe is the set of generated nodes with cost larger than $C$.

\begin{algorithm}[htb!]
\begin{lstlisting}
def $\queryext$($\cost$, budget):
  data.min_fringe = $\infty$ # will be $\lstmathcomment{> C}$
  data.max_visited = 0 # will be $\lstmathcomment{\leq C}$
  data.expanded = 0 # number of expansions
  try:
    graph_search($\cost$, budget, data)
  catch "budget exceeded":
    return $[0, \texttt{data.max\_visited}]$
  return $[\texttt{data.min\_fringe}, \infty]$

# make_node: parent, state, g_cost -> node
def graph_search($C$, budget, data):
  q = make_priority_queue($g$) # $\lstmathcomment{g}$-cost ordering
  insert(q, make_node($\emptyset$, $\sinit$, 0)) #$\lstmathcomment{g(\sinit)=0}$
  visited = {$\sinit$}
  while not empty(q):
    node = extract_min(q)
    s = node.state
    if $s$ in visited # already visited
      continue      # at lower $\lstmathcomment{g}$-cost
    visited += {$s$}
    data.max_visited =
      max(data.max_visited, node.g + $h(s)$)
    if is_goal(s): # optimal solution found to
      throw node # be dealt with by the user
    if data.expanded >= budget:
      throw "budget exceeded"
    data.expanded++
    for $s' \in \succf(s)$:
      node$'$ = make_node(node,$s'$,node.g+$c(s, s')$)
      h$'$ = node$'$.g + $h(s')$
      if h$' \leq C$:
        insert(q, node$'$)
      else:
        data.min_fringe = 
          min(data.min_fringe,h$'$)
  return "no solution"
\end{lstlisting}
\caption{Graph search (simple version)}
\label{alg:graph_search}
\end{algorithm}

\paragraph{Enhancement.}
When the heuristic is sufficiently consistent, A* has an optimal behaviour.
In practice, it seldom happens that an inconsistent heuristic leads to a bad behaviour of A*.
Therefore, to avoid the (small) overhead of \ibbsgs{} for such cases, we propose the following rule of thumb:
Run A* for at least 1000 node expansions. 
Thereafter, if the number of state re-expansions ever becomes at least half of the total number of expansions,
switch to \ibbsgs{}.

\section{Depth-First Search: Further Enhancements}

There are a number of enhancements to DFS that strictly reduce the number of expansions.
\begin{itemize}[leftmargin=*]
    \item \textit{Upper bounds from sub-optimal solutions.}\,\, If a solution is found but the budget is exceeded, then keep the solution cost as an upper bound for subsequent calls.
    \item Detection of duplicate states can be performed along the current trajectory to avoid loops in the underlying graph, while keeping a memory that grows only linearly with the depth of the search, but is now a multiple of the state size.
\end{itemize}
\todo{Mention name 'early stopping' in main text?}

\section{Enhanced DovIBEX Algorithm}\label{sec:zzz_fast}
An enhanced version of \cref{alg:zzzprog} is provided in \cref{alg:zzz_fast}.
The call $\queryp(C, b)$ is like $\query(C, b)$ but additionally returns the number of node expansions.
The optimized version of the algorithm removes programs $k$ from the scheduler if it can be proven that they cannot find a solution.
This pruning happens on line \cref{zzz_fast:prune}, which condition can be fulfilled 
in several circumstances.

Let $k$ be the current program index, then
if a program $k' > k$ has already made a call to query for which the budget was sufficient and the number of nodes expanded was at least $B(k)$, 
then program $k$ can never find a solution and is thus removed from the scheduler.
Observe that if for program $k$ we have $\costhigh \leq \costlow$,
then we necessarily have $b < \budget_{\low}$, since
$b < n(\costhigh) \leq \costlow = \budget_{\low}$.

On \cref{zzz_fast:infty} we know that $\costlow$ is a lower bound
on the cost of the solution, so it is safe to use infinite budget.
This may lead to further pruning on \cref{zzz_fast:prune}.

As for \ibbs{}, \cref{alg:zzz_fast} also has two parameters.
The first one controls the budget, the second one whether we use
an additive exponential search, or a multiplicative one.

With $B(k) = \alpha^k$ and $T(k, r)=r2^k$, it can be worked out that
$T_{\ubs}(k, r) \leq \frac{2}{\alpha-2}\alpha^{k}r^{\log_2 \alpha}$,
where $k = \ceil{\log_\alpha \nstar}$ and $r=r_2$ as in \cref{thm:zzz},
leading
to a number of expansions bounded by $\frac{2\alpha}{\alpha-2}\nstar r_2^{\log_2\alpha}$.
This is more efficient than the default setting when $r$ is small, but becomes
rapidly less efficient for larger $r$.
But note that this can be mitigated by \cref{eq:z3_minb} and possibly by \cref{thm:zzz2}.

\begin{algorithm}[htb!]
\begin{lstlisting}
# T: cost function for UBS
def T(k, r, C): return r$2^\text{k}$

# $\lstmathcomment{\alpha}$: Budget parameter
# is_additive: use additive exp-search?
def DovIBEX_enhanced($\alpha=8$, is_additive):
  # $\textcolor{lstcommentcolor}{\budget_{\low}}$ and $\textcolor{lstcommentcolor}{\costlow}$ are globals, but not $\textcolor{lstcommentcolor}{\costhigh}$
  $\costlow$ = $\cost_{\min}$ # lower bound on $\textcolor{lstcommentcolor}{\coststar}$
  $\budget_{\low}$ = 0 # lower bound on $\textcolor{lstcommentcolor}{\nstar}$
  q = make_priority_queue(T)
  q.insert((1, 1, $\infty$))  # (k=1, r=1, $\textcolor{lstcommentcolor}{\costhigh=\infty}$)
  
  while not q.empty():
    (k, r, $\costhigh$) = q.extract_min()
    b = $\alpha^k$  # budget
    
    if r == 1:
      q.insert((k+1, r, $\infty$))
    
    if b $\leq \budget_{\low}$ or $\costhigh \leq \costlow$: (*\label{zzz_fast:prune}*)
      # Can't find a solution with this budget
      continue  # remove $\textcolor{lstcommentcolor}{k}$ from the queue

    if r == 1:  # IDA* trick, can be omitted (*\label{zzz_fast:ida}*)
      $\cost$ = $\costlow$
      b = $\infty$ (*\label{zzz_fast:infty}*)
    else if $\costhigh$ == $\infty$: # exponential phase
      if is_additive:
        $\cost = \costlow + 2^{r-1}$
      else:
        $\cost$ = $\costlow\times 2$
    else: # binary phase
      $\cost$ = $(\costlow + \costhigh)/2$
    
    $\costlow'$, $\costhigh'$, $\budget_{\low}'$ = $\queryp$($\cost$, b)
    $[\costlow, \costhigh]$ = $[\costlow, \costhigh]\cap[\costlow', \costhigh']$
    
    if $\costhigh'$ == $\infty$:  # budget not exceeded 
      # Solution requires more than $\textcolor{lstcommentcolor}{\budget_{low}'}$ 
      $\budget_{\low}$ = $\budget_{\low}'$
    
    q.insert((k, r+1, $\costhigh$)) 
\end{lstlisting}
\caption{The \ibbslong{} algorithm with a few enhancements}\label{alg:zzz_fast}
\end{algorithm}

Further enhancements can be considered:
\begin{itemize}
    \item If a program $k$ has found an upper bound on the cost,
    then this bound can be propagated to all $k' < k$.
    \item If a solution has been found but the budget is exceeded, keep the cost of the solution as a global upper bound
    (global branch and bound).
\end{itemize}

\section{Worst-Case Re-Expansions for B and B'}\label{sec:mero}

Figure \ref{fig:mero} is a parameterized adaptation of an example from \cite{mero1984pathmax} published by \cite{sturtevant2008using} where A*, B and B' all perform $O(d^2)$ expansions. Note that \cite{sturtevant2008using} also introduce the Delay algorithm, however this algorithm has a hidden assumption which is not necessarily true---namely that any state not counted as part of $\nstargraph$ does not have a shorter path to a state counted as part of $\nstargraph$. As this is not true in general, the Delay algorithm is not completely general.

\begin{figure}
    \centering
\begin{tikzpicture}[font=\scriptsize,xscale=0.9,thick]
\tikzstyle{n} = [draw,inner sep=1pt,minimum width=0.5cm,minimum height=1em]
\tikzstyle{a} = [-latex]
\tikzstyle{l} = [fill=white,inner sep=2pt]
\tikzstyle{l2} = [pos=0.3,fill=white,inner sep=2pt]

\node[n] (s) at (0,0) {$s$};
\node[n] (t1) at (-3,-1.5) {$t_1$};
\node[n] (t2) at (-1.5,-1.5) {$t_2$};
\node[n] (t3) at (0,-1.5) {$t_3$};
\node[n] (t4) at (1.5,-1.5) {$t_4$};
\node[] (t5) at (3,-1.5) {$\cdots$};
\node[n] (t6) at (4.5,-1.5) {$t_d$};
\node[n] (m) at (0,-3) {$m$};
\node[n] (b1) at (-3,-4.5) {$b_1$};
\node[n] (b2) at (-1.5,-4.5) {$b_2$};
\node[n] (b3) at (0,-4.5) {$b_3$};
\node[] (b4) at (1.5,-4.5) {$\cdots$};
\node[n] (b5) at (3,-4.5) {$b_{d-1}$};
\node[n] (b6) at (5.5,-4.5) {$g$};
\draw[a] (s) edge node[l] {1} (t1);
\draw[a] (s) edge node[l] {1} (t2);
\draw[a] (s) edge node[l] {1} (t3);
\draw[a] (s) edge node[l] {1} (t4);
\draw[a] (s) edge node[l] {1} (t6);
\draw[a] (t1) edge node[l2] {$d$} (m);
\draw[a] (t2) edge node[l2] {$d-1$} (m);
\draw[a] (t3) edge node[l2] {$d-2$} (m);
\draw[a] (t4) edge node[l2] {$d-3$} (m);
\draw[a] (t6) edge node[l2] {$1$} (m);
\draw[a] (m) -- (b1);
\draw[a] (b1) edge node[l,pos=0.4] {1} (b2);
\draw[a] (b2) edge node[l,pos=0.4] {1} (b3);
\draw[a] (b3) edge node[l,pos=0.4] {1} (b4);
\draw[a] (b4) edge node[l,pos=0.4] {1} (b5);
\draw[a] (b5) edge node[l,pos=0.5] {$d-1$} (b6);
\draw [ultra thick,decorate,decoration={brace,amplitude=3pt}] (5,-1.2) -- (5,-1.8);
\draw [ultra thick,decorate,decoration={brace,amplitude=4pt}] (5.5,-5) -- (-3,-5);
\node[anchor=west] at (5.1,-1.5) {$d$ states};
\node[anchor=north] at (1.25,-5.1) {$d$ states};
\end{tikzpicture}
    \caption{Worst-case example adapted from \protect\cite{mero1984pathmax} by \protect\cite{sturtevant2008using}.}
    \label{fig:mero}
\end{figure}

\section{Intractability of Admissible Heuristics}\label{sec:intract}

Let $c^*(s', s)$ be the minimum cost of any path from $s'$ to $s$
with $c^*(s, s) = 0$.
Assuming that the heuristic $h$ is admissible, define 
\begin{align*}
    \hat{h}(s) = \max_{s'\in\cS} \{h(s') - c^*(s', s)\}\,,
\end{align*}
which is the best heuristic value that can be propagated forward to $s$.
Note that $\hat{h}$ is consistent.

Let $\coststar$ be minimum cost of any solution state.
Let $n_{\text{opt}} = |\{s: g^*(s) + \hat{h}(s) \leq \coststar\}|$,
which is an upper bound on the number of states that A* expands when using $\hat{h}$ instead of $h$.
The following theorem shows that no algorithm having only access to $h$ can hope to expand $O(n_{\text{opt}})$ without oracle access to the consistent heuristic $\hat{h}$. 

\begin{theorem}\label{thm:intract}
For each deterministic search algorithm there exists a graph search problem with an admissible heuristic such that
the algorithm expands $\Omega(2^{n_{\text{opt}}})$ nodes.
\end{theorem}

\begin{figure}
    \centering
    \begin{tikzpicture}[thick,scale=0.5,font=\tiny]
\tikzstyle{n} = [draw,circle,inner sep=1pt,minimum width=0.3cm]
\tikzstyle{a} = [-latex]
\tikzstyle{b} = [latex-]
\tikzstyle{l} = []
\tikzstyle{s} = [shorten >=25pt,shorten <=25pt,dotted]
\tikzstyle{sw} = [white,ultra thick,shorten >=20pt,shorten <=20pt]
\tikzstyle{br} = [bend right,out=-15]
\tikzstyle{bl} = [fill=white,inner sep=3pt]
\tikzstyle{g} = [fill=red]
\tikzstyle{r} = [red]
\tikzstyle{gr} = [fill=black!20!white]

\node[n,initial,initial where=above,initial text=] (root) at (4.75,2) {$\sinit$};

\node[n] (t11) at (0,0) {$s_2$};
\node[n,g] (t21) at (-2,-2) {};
\node[n,gr] (t22) at (2,-2) {};
\node[n,g] (t31) at (-3,-4) {};
\node[n,gr] (t32) at (-1,-4) {};
\node[n,g] (t33) at (1,-4) {};
\node[n,gr] (t34) at (3,-4) {};
\node[n,gr] (t41) at (-3.5,-6) {};
\node[n,g] (t42) at (-2.5,-6) {$s\smash{^*}$}; 
\node[n,gr] (t43) at (-1.5,-6) {};
\node[n,g] (t44) at (-0.5,-6) {};
\node[n,gr] (t45) at (0.5,-6) {};
\node[n,g] (t46) at (1.5,-6) {};
\node[n,gr] (t47) at (2.5,-6) {};
\node[n,g] (t48) at (3.5,-6) {};

\node[anchor=north] at (-2.5,-6.3) {goal};

\node[n,g] (m1) at (4,-2) {};
\node[n,gr] (m2) at (5,-3) {};
\node[n,g] (m3) at (6,-4) {};
\node[n,gr] (m4) at (7,-5) {};
\node[n,gr] (m5) at (8,-6) {};
\node[n,g] (m6) at (9,-7) {};

\node[n,gr] (n1) at (6,0) {$l_1$};
\node[n,gr] (n2) at (7,-1) {$r_1$};
\node[n,gr] (n3) at (8,-2) {$l_2$};
\node[n,gr] (n4) at (9,-3) {$r_2$};
\node[n,gr] (n5) at (10,-4) {$l_3$};
\node[n,gr] (n6) at (11,-5) {$r_3$};

\node[n] (s) at (9.5,-1.5) {$s_1$};
\node[anchor=west,text width=1.5cm] (sl) at (9.8,-1.5) {$h(s_1)\!=\!2^d\!+\!d\!+\!3$};

\draw[a] (root) edge[out=0,in=90] node {} (s);
\draw[a] (root) edge[out=-180,in=90] node {} (t11);

\draw[a] (s) edge node {} (n1);
\draw[a] (s) edge node {} (n2);
\draw[a] (s) edge node {} (n3);
\draw[a] (s) edge node {} (n4);
\draw[a] (s) edge node {} (n5);
\draw[a] (s) edge node {} (n6);

\draw[a] (n1) edge (m1);
\draw[a] (n2) edge (m2);
\draw[a] (n3) edge (m3);
\draw[a] (n4) edge (m4);
\draw[a] (n5) edge (m5);
\draw[a] (n6) edge (m6);

\draw[sw] (n1) edge node[scale=2] {\color{red} \rotatebox{45}{$\bm{\times}$}} (m1);
\draw[sw] (n2) edge (m2);
\draw[sw] (n3) edge node[scale=2] {\color{red} \rotatebox{45}{$\bm{\times}$}} (m3);
\draw[sw] (n4) edge (m4);
\draw[sw] (n5) edge (m5);
\draw[sw] (n6) edge node[scale=2] {\color{red} \rotatebox{45}{$\bm{\times}$}} (m6);

\draw [ultra thick,decorate,decoration={brace,amplitude=4pt}] (11.4,-5.4) -- (9.4,-7.4);
\node[anchor=west] at (10.4,-6.8) {$2^d$ nodes};

\draw[s] (n1) edge (m1);
\draw[s] (n2) edge (m2);
\draw[s] (n3) edge (m3);
\draw[s] (n4) edge (m4);
\draw[s] (n5) edge (m5);
\draw[s] (n6) edge (m6);

\draw[b,r] (t21) edge[br] (m1);
\draw[b] (t22) edge[br] (m2);

\draw[b,r] (t31) edge[br] (m3);
\draw[b] (t32) edge[br] (m4);
\draw[b,r] (t33) edge[br] (m3);
\draw[b] (t34) edge[br] (m4);

\draw[b] (t41) edge[br] (m5);
\draw[b,r] (t42) edge[br] (m6);
\draw[b] (t43) edge[br] (m5);
\draw[b,r] (t44) edge[br] (m6);
\draw[b] (t45) edge[br] (m5);
\draw[b,r] (t46) edge[br] (m6);
\draw[b] (t47) edge[br] (m5);
\draw[b,r] (t48) edge[br] (m6);

\draw[a] (t11) edge node[l] {} (t21);
\draw[a] (t11) edge node[l] {} (t22);
\draw[a] (t21) edge node[l] {} (t31);
\draw[a] (t21) edge node[l] {} (t32);
\draw[a] (t22) edge node[l] {} (t33);
\draw[a] (t22) edge node[l] {} (t34);
\draw[a] (t31) edge node[l] {} (t41);
\draw[a] (t31) edge node[l] {} (t42);
\draw[a] (t32) edge node[l] {} (t43);
\draw[a] (t32) edge node[l] {} (t44);
\draw[a] (t33) edge node[l] {} (t45);
\draw[a] (t33) edge node[l] {} (t46);
\draw[a] (t34) edge node[l] {} (t47);
\draw[a] (t34) edge node[l] {} (t48);
\end{tikzpicture} 
    \caption{The graph for the proof of \cref{thm:intract} with $d=3$.
    The large heuristic value of $s_1$ propagates through all the gray nodes.
    Red crosses in the chains indicate that the chain contains an edge with infinite cost, preventing the large heuristic from being propagated to the full binary tree.
    Only one path in the tree is entirely red, and leads to the goal g.}
    \label{fig:intract}
\end{figure}

\begin{proof}
Let $d$ be an arbitrarily large integer.
Consider the following class of problems, illustrated in \cref{fig:intract}.
The initial state is $\sinit$, which has $g(\sinit) = 0$ and successors
$\succf(\sinit) = \{s_1, s_2\}$.
Below $s_2$ is a full binary tree of depth $d$.
State $s_1$ has successors $\succf(s_1) = \{l_1,r_1,\ldots,l_d,r_d\}$, each of which starts a chain of length $2^d$ with the final state in the chains started by $l_m$ (respectively $r_m$) connected to all left-hand (respectively right-hand) children in the binary tree at depth $m$.
The heuristic values are 0 except $h(s_1) = 2^d + d + 3$.
Consider the first $2^d-1$ expansions of the search algorithm where all edge costs are unitary
and there is no goal state. By the pigeonhole principle there exists a state $s$ in the leaves of the binary tree that has not been expanded.
We now modify the graph so that
(a) the behavior of the algorithm is identical
(b) the goal is in state $s^*$ and (c) $n_{\text{opt}} = d + 2$.
To do this, let $s^*$ be the goal state and $\smash{(a_t)_{t=1}^{d+2}}$ be the unique path ending in $s^*$ and passing through $s_2$.
Then for each $3 \leq m \leq d+2$ find an edge in the chain connected to state $a_m$ that was not examined by the algorithm in the goal-less graph and set its cost to infinity, which cuts the path from state $s_1$ to state $s_m$.
Since the algorithm has not examined this edge, it cannot prove whether the heuristic value of $s_1$ should be propagated to the corresponding node in the full binary tree.
The large heuristic in state $s_1$ ensures that children in the binary tree that do not lead to the goal state have a heuristic of at least $d+3$,
whereas all states along the path $a$ have heuristic 0.
Hence $n_{\text{opt}} = d+2$.
Finally, by construction the algorithm expands at least $2^d-1$ nodes in the modified graph before finding the goal. The result follows since $d$ may be chosen arbitrarily large.
\end{proof}

\begin{remark}
The example is complicated by the fact that we did not assume the algorithm was restricted to only call the successor function on previously expanded nodes. We only used that an edge-cost is observed when its parent is expanded. We did not even assume that the algorithm is guaranteed to return an optimal solution.
The proof is easily modified to lower bound the expected number of expansions by any randomized algorithm using Yao's minimax principle and by randomizing the position of the goal in the leaves of the binary tree and the infinite-cost edges in the chains.
\end{remark}

\begin{remark}
\cref{thm:intract} also holds when using 
pathmax \cite{mero1984pathmax} or BPMX \cite{felner2005dual},
since in order to make the heuristic consistent the algorithm still needs to expand
exponentially many nodes along the chains.
\end{remark}

\section{Proofs of \texorpdfstring{\cref{prop:es-game-end,prop:es-nobudget}}{}}



In the following we use
$\eps = \texttt{start}$, $h = \texttt{high}$ and $l = \texttt{low}$.

\begin{proof}[Proof of \texorpdfstring{\cref{prop:es-game-end}}{}]
We know that $\coststar < \costcrit(b)$ by definition of $b$ through $k$,
and that $\eps \in [\coststar, \costcrit(b))$.

The number of queries in the exponential phase before $h \geq \coststar$ is at most 
\begin{align*}
    \min \{k\in\Naturals: \eps 2^k \geq \coststar\} = \ceil{\log_2  \frac{\coststar}{\eps}}_{\geq 1}\,.
\end{align*}
At the end of the exponential phase,
if $h \in [\coststar, \costcrit(b))$ then the game ends.
Otherwise $h \geq \costcrit(b)$,
and also $l \leq \coststar$ (otherwise a game-ending query would have been made;
$l = \coststar$ is possible iff $\eps = \coststar$).
Thus $l \leq \coststar < \costcrit(b) \leq h$,
and since $h = 2l$ we have $h-l \leq \coststar$.
During the binary phase, the size of the interval $[l, h]$ at least halves after each query.
Now, suppose that at some point
$h - l < 2(\cost_{\crit}(b) - \coststar)$.
Using $C= (h+l)/2$
then we obtain $C < l + \cost_{\crit}(b) - \coststar \leq \cost_{\crit}(b)$
since $l \leq \coststar$,
and also $C > h - (\cost_{\crit}(b)-\coststar) \geq \coststar$ 
since $h \geq \cost_{\crit}(b)$.
Therefore $C\in(\coststar, \costcrit(b))$
which is a game-ending query.
Starting from $h-l \leq \coststar$, this requires at most
\begin{align*}
    1+ \min&\{k\in\Nonnegints: \coststar/2^k \leq 2(\costcrit(b) - \coststar)\} \\
    &= 1+ \floor{1+\log_2 \frac{\coststar}{2(\costcrit(b)-\coststar)}}_{\geq 0} \\
    &= 1 + \floor{\log_2 \frac{\coststar}{\costcrit(b)-\coststar}}_{\geq 0}
\end{align*}
calls to query before ending the game.
Therefore the number of calls to query
is at most $\nbisec(\texttt{start}, \coststar, \costcrit(b) - \coststar)$.
\end{proof}

\begin{proof}[Proof of \cref{prop:es-nobudget}]
Let $x < y < z$ be three consecutive values in $A$ 
(that is, $y = \tceil{x}$ and $z = \tceil{y}$)
with $y = \costcrit(b)$.
During the exponential phase, the number of queries until $h \geq \costcrit(b)$ is at most
\begin{align*}
    \min\{k\in\Naturals: \eps2^k \geq \costcrit(b)\} &= \ceil{\log_2 \frac{\costcrit(b)}{\eps}}_{\geq 1} \\
    &\leq 
    \ceil{\log_2 \frac{\coststar}{\eps}}_{\geq 1}\,.
\end{align*}
At the end of this phase, we have $l \leq \costcrit(b) \leq h$
and since $h=2l$ we have $h-l \leq \costcrit(b)$.
Now for the binary phase, where the interval $[l, h]$ is always
at least halved.
Observe that $\queryext$ ensures that if $x \leq \cost < y$ then $l$ is set to $y$,
and if $y \leq \cost < z$ then $h$ is set to $y$ too.
Next we show that if $h - l < \fgapmin$ 
then $h < z$ and $l > x$:
\begin{align*}
\tag{by assumption} h - l &< \fgapmin \leq z-y\\
\tag{using $l \leq y$} h &< z-y+l \leq z \\
\tag{by assumption} h - l &< \fgapmin \leq y-x\ \\
\tag{using $h \geq y$} l &> x-y+h \geq x \,. 
\end{align*}
(Remembering that $y = \costcrit(b)$,
observe that if $\texttt{start}=y$ then $l = y$ already at the beginning of \cref{alg:es}.)
Thus  $h = l = y$ which terminates the algorithm and returns $l = \costcrit(b)$.

Hence the number of calls to query during the binary phase before $h-l < \fgapmin$ (which entails
$h=l$) is at most
(remembering that $h-l \leq \costcrit(b)$ at the end of the exponential phase)
\begin{align*}
    &\min\{k\in\Nonnegints: \costcrit(b)/2^k < \fgapmin\} \\
    &\qquad\qquad= \floor{1+ \log_2 \frac{\costcrit(b)}{\fgapmin}}_{\geq 0} \\
    &\qquad\qquad\leq 1+ \floor{\log_2 \frac{\coststar}{\fgapmin}}_{\geq 0}\,.
\end{align*}
Therefore, the total number of queries is at most 
$\nbisec(\texttt{start}, \costcrit(b), \fgapmin)$.
\end{proof}

\section{Proofs of \texorpdfstring{\cref{thm:ubs} and \cref{cor:ubs}}{}}

\begin{proof}[Proof of \cref{thm:ubs}]
Suppose when \ubs{} is executing $(k, \seg)$, the node
$(j, m)$ has already been executed by \ubs{}.
Then the optimality property of Uniform Cost Search~\cite{russell2009aima} ensures that $T(j, m) \leq T(k, \seg)$.
Hence, program $j$ has been executed for at most $T_j(k, \seg)$ steps.
The result follows by summing over all programs and using that program $j$ never runs for more than $\tau_j$ steps.
\end{proof}

\begin{proof}[Proof of \cref{cor:ubs}]
The result follows from \cref{thm:ubs} and the assumptions that $T(j, 0) = 0$ and $T(j, 1) \leq T(j, m)$ for $m \geq 1$.
\end{proof}

\section{Proofs of \texorpdfstring{\cref{thm:zzz,thm:zzz2}}{}}

\begin{proof}[Proof of \cref{thm:zzz}]
Let $k^* = \ceil{\log_2(\nstar)}$.
By \cref{prop:es-game-end}, 
programs with $k \geq k^*$ make a game-ending query after at most
\begin{align*}
\seg_k = \nbisec(\cost_{\min},\ \ \coststar,\ \ \cost_{\crit}(b_k) - \coststar)
\end{align*}
calls to \code{run_prog($k$, $2^k$)}.
By \cref{thm:ubs}, \cref{alg:zzzprog} does not use more than
$\min_{k \geq k^*} \seg_k 2^k\floor{k + \log_2 \seg_k}$
steps before exiting with a solution, which also upper bounds the number of expansions of the algorithm. 
The first bound is obtained by taking $k = k^*$ for which $b_{k^*} \leq 2 \nstar$ and
$\cost_{\crit}(2^{k_*}) \geq \cost_{\crit}(\nstar) = \coststar + \fgap(\coststar)$.
The second bound is obtained by noting that 
for any $b\geq \nstar$ there exists a $k$ such that
$2^k \geq \nstar$ and
$\seg_k 2^k \leq 2b\nbisec(\cost_{\min}, \coststar, \cost_{\crit(b)}-\coststar)$.
\end{proof}

\begin{proof}[Proof of \cref{thm:zzz2}]
Similarly to the proof of \cref{thm:rv},
let $k_* = \ceil{\log_2 (\nstar)}$  be the first program with enough budget,
that is $2^{k_*} \geq \nstar$.
Using \cref{prop:es-game-end},
program $k_*$ terminates after at most $\seg_2$ calls to $\query$.
Each program $k < k_*$ requires at most $\seg_1$
calls to $\query$ to terminate, that is $\tau_k \leq 2^k \seg_1$.
Each program $k > k_*$ that has started when $k_*$ terminates
has used at most $T_k(k_*, \seg_2)$ steps,
and only programs $k\leq k_* + \floor{\log_2 \seg_2}$ (that is $2^k \leq \seg_2 2^{k_*}$)
have started (that is, $T_k(k_*, \seg_2) > 0$).
Hence, using \cref{thm:ubs}, the number of steps is bounded by
\begin{align*}
    T_{\ubs}(k_*, \seg_2) &\leq \sum_{k < k_*} \tau_k + \sum_{k \geq k_*} T_k(k_*, \seg_2) \\
    &\leq \sum_{k=1}^{k_*-1} 2^k \seg_1 + \sum_{k=k_*}^{k_* + \floor{\log_2 n_2}} 2^{k_*}\seg_2 \\
    &= (2^{k_*}-1)\seg_1 + (1+\floor{\log_2 \seg_2})2^{k_*}\seg_2 \\
    &\leq 2\nstar(\seg_1 + \seg_2(1+\floor{\log_2 \seg_2}))\,,
\end{align*}
which also bounds the number of expansions of the algorithm.
\end{proof}

\fi

\ifdefined\nocomment{}\else
\clearpage
\input{todo} 
\fi

\end{document}